%% file: main.tex
\documentclass[12pt]{article}
\usepackage{amsmath}
\usepackage{graphicx,psfrag,epsf}
\usepackage{enumerate}
\usepackage{natbib}
\usepackage{url} 

\newcommand{\blind}{1}

\usepackage{amsmath,amssymb,amsthm,amsfonts}
\usepackage{bm}
\usepackage{natbib}
\usepackage{geometry}
\usepackage{hyperref}
\usepackage{bbm}
\usepackage{tikz}
\usetikzlibrary{positioning}
\usepackage{enumitem}
\usepackage{changepage}
\usepackage{subcaption}

\newtheorem{assumption}{Assumption}[section]
\newtheorem{theorem}{Theorem}

\newtheorem{definition}{Definition}
\newtheorem{remark}{Remark}

\newcommand{\indep}{\perp\!\!\!\perp}
\newcommand{\E}{\mathbb{E}}
\newcommand{\covbounds}{[L(y_1, y_0), U(y_1, y_0)]}
\newcommand{\margbounds}{[L_\text{marg}(y_1, y_0), U_\text{marg}(y_1, y_0)]}
\newcommand{\1}{\mathbbm{1}}
\newcommand{\Var}{\mathrm{Var}}
\newcommand{\IF}{\mathrm{IF}}
\newcommand{\Pn}{\mathbb{P}_n}
\renewcommand{\P}{P} 
\newcommand{\dto}{\overset{d}{\to}}
\newcommand{\convp}{\overset{p}{\to}}
\newcommand{\N}{\mathcal{N}}

\begin{document}

\def\spacingset#1{\renewcommand{\baselinestretch}%
{#1}\small\normalsize} \spacingset{1}


\if1\blind
{
  \title{\bf Nonparametric Identification and Inference for Counterfactual Distributions with Confounding}
  \author{Jianle Sun $^1$ \& Kun Zhang $^{1,2}$
    \\
    $^1$ Department of Philosophy, Carnegie Mellon University\\
    $^2$ Machine Learning Department,\\Mohamed bin Zayed University of Artificial Intelligence}
  \maketitle
} \fi

\if0\blind
{
  \bigskip
  \bigskip
  \bigskip
  \begin{center}
    {\LARGE\bf Nonparametric Identification and Inference for Counterfactual Distributions with Confounding}
\end{center}
  \medskip
} \fi

\bigskip
\begin{abstract}
    We propose nonparametric identification and semiparametric estimation of joint potential outcome distributions in the presence of confounding. First, in settings with observed confounding, we derive tighter, covariate-informed bounds on the joint distribution by leveraging conditional copulas. To overcome the non-differentiability of bounding min/max operators, we establish the asymptotic properties for both a direct estimator with polynomial margin condition and a smooth approximation with log-sum-exp operator, facilitating valid inference for individual-level effects under the canonical rank-preserving assumption. Second, we tackle the challenge of unmeasured confounding by introducing a causal representation learning framework. By utilizing instrumental variables, we prove the nonparametric identifiability of the latent confounding subspace under injectivity and completeness conditions. We develop a ``triple machine learning" estimator that employs cross-fitting scheme to sequentially handle the learned representation, nuisance parameters, and target functional. We characterize the asymptotic distribution with variance inflation induced by representation learning error, and provide conditions for semiparametric efficiency. We also propose a practical VAE-based algorithm for confounding representation learning. Simulations and real-world analysis validate the effectiveness of proposed methods. By bridging classical semiparametric theory with modern representation learning, this work provides a robust statistical foundation for distributional and counterfactual inference in complex causal systems.
\end{abstract}

\noindent%
\vfill

\newpage
\spacingset{1.45} 

\section{Introduction}

Causal inference fundamentally aims to predict how individuals or populations respond to competing interventions, thereby concerning the comparison of potential outcomes under alternative treatment regimes. While classical estimands such as the Average Treatment Effect (ATE) focus on mean differences, many scientifically relevant questions, such as the probability of benefit, quantile effects, or distributional shifts, depend on the entire distributions of potential outcomes. However, researchers typically face a dual hurdle in capturing these distributions. First, in the presence of confounding, even the marginal distributions of $Y(1)$ and $Y(0)$ are generally not identifiable, and existing instrumental variable (IV) approaches often rely on restrictive parametric assumptions or focus only on local effects \cite{angrist1996identification,swanson2018partial}. Second, even when conditional ignorability holds and marginals are identifiable, the joint distribution $(Y(1), Y(0))$ remains fundamentally unobservable without additional structural assumptions.

This paper bridges these gaps by providing a unified, principled framework for distributional causal inference. Our first contribution addresses the ``missing data" problem of the joint distribution under the assumption of no unmeasured confounding. We develop tight, covariate-informed Fr\'echet-Hoeffding (FH) bounds \cite{nelsen2006introduction} on the joint distribution under no unmeasured confounding. Leveraging conditional copulas, we show that the sharp upper bound admits a clear structural interpretation as \emph{conditional rank preservation} (or conditional comonotonicity) \cite{nelsen2006introduction}, a canonical assumption underlying individual treatment effect estimation and counterfactual reasoning \cite{xie2023advancing,wu2025learning}. To move from theory to practice, we address the non-smoothness of these bounds via two complementary paths: a direct estimator under a polynomial margin condition and a smooth log-sum-exp approximation. We further establish their asymptotic properties, enabling valid frequentist inference and confidence intervals for rank-preserving structures \cite{levis2025covariate}.

Our second contribution tackles the more daunting scenario where confounding is unmeasured, and in this scenario even the non-parametric identification for marginal distributions is non-trivial. Inspired by recent advances in causal representation learning \cite{kong2022partial,ng2025debiasing,Moran09022026}, we propose a representation learning based framework that leverages IVs to recover latent confounding structures. Under suitable completeness and independence assumptions, we show that the confounding subspace is identified up to an invertible transformation. This allows the learned representation to serve as a valid proxy for unobserved confounders, thereby enabling identification of marginal potential outcome distributions in complex settings. 

To implement this framework, we introduce a \emph{triple machine learning} (TML) procedure that extends double machine learning \cite{chernozhukov2018double,kennedy2024semiparametric} by incorporating an additional cross-fitting stage for representation learning. We rigorously characterize the impact of first-stage representation error on the variance inflation in asymptotic distribution, identifying regimes in which super-convergence of the representation learner yields semiparametric efficiency.

For practical estimation on confounding representations, we propose an Instrumental Variable Variational Autoencoder (IV-VAE) augmented with a Hilbert-Schmidt Independence Criterion (HSIC) penalty \cite{gretton2005measuring} to ensure that the recovered latent factors are truly exogenous to the instrument, satisfying the core identifying assumptions. Rather than explicitly modeling instrument-dependent latent factors, we adopt a reduced-form design that conditions the decoder directly on the observed instrument, allowing instrument-induced variation to be absorbed while isolating latent confounding structure.

The remainder of the paper is organized as follows. Section \ref{sec:bound-cov} details the identification of rank-preserving bounds via conditional copulas and presents the asymptotic theory for the proposed estimators. Section \ref{sec:marginal-iv} introduces the representation learning framework for unmeasured confounding, derives the properties of the triple machine learning estimator, and proposes an effective VAE-based learning approach. Section \ref{sec:bound-iv} discusses the synthesis of these methods. Section \ref{sec:sim} provides simulation results. Section \ref{sec:app} applies our proposed methods to analyze the demand for cigarettes in US. Proofs and technical details are deferred to the Appendix.

\section{Bound joint counterfactuals with all confounders observed via conditional copulas}\label{sec:bound-cov}
When there is no unobserved confounding, the identification of counterfactual marginals $F_{Y(a)}$ is easy. We will briefly review the identification results for the marginal distributions of potential outcomes and discuss leveraging the conditional copulas of the observed confounding covariates to derive tighter bounds for the joint distribution of the potential outcomes. We are particularly interested in the upper bound (which corresponds to the joint distribution under a conditional rank-preserving assumption), as this will provide an important foundation for performing individualized counterfactual inference. However, we must also address the challenges to estimation posed by the non-differentiable $\min / \max$ functions.
\subsection{Identification}\label{sec:notation}
\subsubsection{Identification of counterfactual marginals}
We observe \(n\) i.i.d.\ draws \(O_i=(Y_i,A_i,X_i)\sim \P\), \(i=1,\dots,n\). Here
\(A\in\{0,1\}\) is treatment, \(Y\in\mathbb{R}\) is outcome, and \(X\in\mathcal{X}\subset\mathbb{R}^d\) are observed covariates. Let $Y(a)$ denote the potential outcome (counterfactual outcome) under treatment level $a$, and its marginal can be identified under the standard assumptions.
\begin{definition}[Nuisance Functions and Distributions]
Let $X$ be a vector of covariates (observed confounders).
\begin{enumerate}
    \item Let $\theta_a(X) := F_{Y(a)\mid X}(y\mid X) = P(Y(a) \le y \mid X)$ be the conditional CDF of the potential outcome $Y(a)$ given $X$.
    \item Let $F_{Y(a)}(y) := \E_X[\theta_a(X)]$ be the marginal CDF of $Y(a)$, obtained by integrating the conditional CDF over the distribution of $X$.
\end{enumerate}
\end{definition}

\begin{theorem}[Identification of counterfactual marginals]\label{ass:id}
Assume the standard causal identification conditions
\begin{enumerate}
  \item[(i)] \textbf{Consistency:} $Y = Y(A)$.
  \item[(ii)] \textbf{Conditional ignorability (No unmeasured confounding):} $(Y(1),Y(0)) \perp A \mid X$.
  \item[(iii)] \textbf{Positivity:} $\P(A=a\mid X)>0$ almost surely for $a=0,1$.
\end{enumerate}
\end{theorem}
Under Assumption~\ref{ass:id}, the conditional marginals
$F_{Y(a)\mid X}(y\mid x)$ are identifiable from observed data, 
\[
\theta_a(x) := F_{Y(a)\mid X}(y\mid x) = \P(Y\le y \mid A=a, X=x).
\]
And the marginals
\[
F_{Y(a)}(y) = \E[\theta_a(X)] = \E_X\left[F_{Y(a)\mid X}(y\mid X=x)\right] =\int F_{Y(a)\mid X}(y\mid x)\, dF_X(x).
\]

\subsubsection{Conditional Fr\'echet–Hoeffding bounds on counterfactual joint distributions}
We aim to move beyond counterfactual marginals to the joint distribution. Generally speaking, by Sklar's Theorem, any joint distribution can be represented by its marginals $F_{Y(1)}(y_1)$ and $F_{Y(0)}(y_0)$ and a copula function $C$, such that $F_{Y(1),Y(0)}(y_1, y_0) = C(F_{Y(1)}(y_1), F_{Y(0)}(y_0))$. 

Without further assumptions, the copula $C$ is only bounded by the Fr\'echet-Hoeffding bounds $L(u_1, u_0) \le C(u_1, u_0) \le U(u_1, u_0)$, where $L(u_1, u_0) = \max(u_1+u_0-1, 0)$ (the countermonotonicity copula) and $U(u_1, u_0) = \min(u_1, u_0)$ (the comonotonicity copula). This implies the \emph{marginal bounds} on the joint distribution
\[
\max\{ F_{Y(1)}(y_1) + F_{Y(0)}(y_0) - 1, \, 0 \} \le F_{Y(1),Y(0)}(y_1, y_0) \le \min\{ F_{Y(1)}(y_1), \, F_{Y(0)}(y_0) \}.
\]
We can also bound the joint distribution via conditional copulas, which yields a tighter result. Using a conditional version of Sklar's theorem, we can write $F_{Y(1),Y(0)\mid X}(y_1,y_0\mid x) = C_x(\theta_1(x), \theta_0(x))$, where $\theta_a(x) = F_{Y(a)\mid X}(y_a\mid x)$ and $C_x$ is a conditional copula that may depend on $x$. For any $x$, this copula is bounded by the same Fr\'echet--Hoeffding limits
\[
\max\{ \theta_1(x)+\theta_0(x)-1,\,0\}
\le F_{Y(1),Y(0)\mid X}(y_1,y_0\mid x)
\le \min\{\theta_1(x),\theta_0(x)\}.
\]
Integrating over $X$ (by the law of total expectation) gives the unconditional bounds
\[
L(y_1,y_0)\le\; F_{Y(1),Y(0)}(y_1,y_0)
\;\le U(y_1,y_0).
\]
where
\begin{align}
    L(y_1,y_0) &\;:=\; \E_X\bigl[\max\{\theta_1(X)+\theta_0(X)-1,\,0\}\bigr],\\
    U(y_1,y_0) &\;:=\; \E_X\bigl[\min\{\theta_1(X),\theta_0(X)\}\bigr].
\end{align}
In particular, the conditional upper bound $M(\theta_1(x), \theta_0(x)) = \min\{\theta_1(x),\theta_0(x)\}$ is achieved under the \emph{conditional rank-preserving} (or conditional comonotonicity) assumption. This assumption states that for each $X=x$, $Y(1)$ and $Y(0)$ are comonotone functions of the same latent rank variable $U \sim \text{Unif}[0,1]$, such that $Y(a) = F_{Y(a)\mid X}^{-1}(U\mid X)$. The integrated upper bound $U(y_1, y_0)$ thus corresponds to assuming this conditional rank-preservation holds across all $x$.

\begin{theorem}\label{thm:bound}
Consider bounds on $F_{Y(1), Y(0)}(y_1, y_0)$
    \begin{enumerate}
    \item The \textbf{covariate-conditional bounds} are derived by first applying the Fr\'echet-Hoeffding bounds conditional on $X$, and then integrating over $X$:
    \begin{align*}
        L(y_1, y_0) &:= \E_X\bigl[\max\{\theta_1(X)+\theta_0(X)-1,\,0\}\bigr] \\
        U(y_1, y_0) &:= \E_X\bigl[\min\{\theta_1(X),\theta_0(X)\}\bigr]
    \end{align*}
    
    \item The \textbf{marginal bounds} are derived by applying the Fr\'echet-Hoeffding bounds directly to the marginal distributions $F_{Y(1)}(y_1)$ and $F_{Y(0)}(y_0)$
    \begin{align*}
        L_\text{marg}(y_1, y_0) &:= \max\{F_{Y(1)}(y_1)+F_{Y(0)}(y_0)-1,\,0\} \\
        U_\text{marg}(y_1, y_0) &:= \min\{F_{Y(1)}(y_1),\,F_{Y(0)}(y_0)\}
    \end{align*}
\end{enumerate}
The bounds on the joint cumulative distribution function (CDF) $F_{Y(1), Y(0)}(y_1, y_0)$ derived from covariate-conditional distributions are always tighter than, or equal to, the bounds derived from the marginal distributions, i.e.,
\[
L_\text{marg}(y_1, y_0) \le L(y_1, y_0) \quad \text{and} \quad U(y_1, y_0) \le U_\text{marg}(y_1, y_0)
\]
This implies that $\covbounds \subseteq \margbounds$.
\end{theorem}
The proof relies on a direct application of Jensen's Inequality (see appendix \ref{proof:bound}). Equality holds if and only if the conditional CDFs $\theta_a(X)$ are constant almost surely with respect to $X$, implying that the covariates $X$ provide no additional information beyond the marginals.

\subsection{Estimation and Inference}
In what follows we present estimation and inference for functional of interest
\begin{equation}\label{eq:theta}
\Psi(P):=\E_{X\sim P_X}\big[\phi\big(\theta_0(X),\theta_1(X)\big)\big].
\end{equation}
where $\phi^U(x,y)=\min\{x,y\}$ and $\phi^{\mathrm{L}}(x,y)=\max\{x+y-1,\,0\}$ correspond to the upper and lower bound, respectively.
We will mainly focus on the upper bound $\Psi(P)=U(y_1,y_0)$ since it corresponds to the rank-preserving joint distribution when conditioning on all covariates, which is plausible and meaningful in many real counterfactual reasoning problems. Results for $L(y_1,y_0)$ can be obtained analogously. 
We will use standard notation: \(\Pn\) is empirical measure, \(\|\cdot\|_{L_2(P)}\) denotes \(L_2(P)\)-norm, \(\|\cdot\|_{P,\infty}\) the essential sup norm, and \(\convp\) convergence in probability. Let $\hat\theta_a(X)$ denote estimators of the nuisance functions $\theta_a(X)$,
for example via conditional CDF regression or flexible machine learning methods with cross-fitting.

A natural way is the plug-in estimator $\hat\Psi_{\mathrm{plug-in}} = \P_n\phi\{\hat\theta_1(X_i),\hat\theta_0(X_i)\}$
while it suffers from the estimation efficiency of nuisance functional $\theta_a(x)$. Instead, we consider the double-robust estimand augmented inverse-propensity weighted (AIPW) form. We first give explicit influence-function components for the conditional CDFs. When \(\phi\) is non-smooth (min or max), the partial derivatives do not exist on the boundary set, and we present (A) the direct (nonsmooth) route under a polynomial margin condition, and (B) the smooth approximation route (log-sum-exp) to handle the issue \cite{levis2025covariate}. 

\subsubsection{Direct Estimation under a Margin Condition}
We want to estimate
\[
\Psi(P)=\E[\min\{\theta_0(X),\theta_1(X)\}] = \E[\theta_{d(X)}(X)].
\]
Write $\pi_a(x)=\P(A=a\mid X=x)$. Define the unknown ``oracle" selector
\[
d(x)=\arg\min_{a\in\{0,1\}}\theta_a(x),\qquad \theta_a(x)=\P(Y\le y\mid A=a,X=x).
\]
and we replace it by the plug-in selector
\[
\hat d(x) = \arg\min_{a}\hat\theta_a(x).
\]

To analyze the asymptotic behaviors of the hybrid estimator, we introduce the following polynomial margin.

\begin{assumption}[Polynomial margin]\label{ass:polymargin}
There exist constants $\alpha>0$ and $C<\infty$ such that for all $t>0$,
\[
\P\big(|\theta_1(X)-\theta_0(X)|\le t\big)\le C t^\alpha.
\]
\end{assumption}

In fact, the parameter $\alpha$ characterizes the separation degree between the two potential outcome distributions at the unit level. Geometrically, it quantifies the probability mass of the covariate space where the two conditional CDFs, $\theta_1(X)$ and $\theta_0(X)$, are nearly identical. This assumption controls the mass of $X$ near the "tie" or "ambiguity" region $\{x:\theta_0(x)=\theta_1(x)\}$. A large $\alpha$ implies a strong margin, where the population is clearly partitioned into regions where one potential outcome is strictly more likely to be below the threshold than the other. On the contrary, a small $\alpha$ signifies a weak margin, indicating a high density of individuals whose conditional ranks are indistinguishable. This creates a noise-sensitive boundary where small estimation errors in $\hat{\theta}_a$ can lead to frequent mis-selection in $\hat{d}(x)$, which we will see result in a more stringent sup-norm convergence rates required to control the bias.

\medskip

\begin{theorem}[Asymptotic properties of margin - direct estimator]\label{thm:direct}
Under Assumption~\ref{ass:polymargin} and the boundedness conditions that there exist constants $0<\underline c<\overline c<\infty$ such that
$\underline c < \hat\pi_a(x) < \overline c$ almost surely, suppose cross-fitting is used (or Donsker conditions hold). Then the estimator 
\[
\hat\Psi = \Pn[\varphi(O;\hat P,\hat d)]
= \Pn\Big[\sum_{a=0}^1 \1\{\hat d(X)=a\} \Big(
\hat\theta_a(X) + \frac{\1\{A=a\}}{\hat\pi_a(X)}(\1\{Y\le y\}-\hat\theta_a(X))\Big)\Big]
\]
\begin{itemize}
    \item (consistency) is consistent $|\hat\Psi-\Psi|=o_p(1)$
    when the nuisance estimators satisfy
    \[
    \max_a \|\hat\theta_a-\theta_a\|_\infty = o_p(1),
    \qquad
    \|\hat\theta_a-\theta_a\|_{L_2(P)}\|\hat\pi_a-\pi_a\|_{L_2(P)}=o_p(1),\ \ a=0,1,
    \]
    \item (root-$n$ consistency and asymptotic normality) satisfies
    \[
    \hat\Psi-\Psi = (\Pn-\P)\varphi(O;P,d) + o_p(n^{-1/2}),
    \]
    hence
    \[
    \sqrt{n}(\hat\Psi-\Psi)\dto N(0,\Var\big(\varphi(O;P,d))\big).
    \]
    when nuisance estimators satisfy
    \[
    \max_a \|\hat\theta_a-\theta_a\|_\infty = o_p\big(n^{-1/(2(1+\alpha))}\big),
    \qquad
    \|\hat\theta_a-\theta_a\|_{L_2(P)}\|\hat\pi_a-\pi_a\|_{L_2(P)}=o_p(n^{-1/2}),\ \ a=0,1.
    \]
\end{itemize}
\end{theorem}
The rate conditions can be established by following the standard semiparametric argument, i.e., the following von-Mises decomposition
\begin{align}
    \hat\Psi - \Psi(P)
    &= \Pn[\varphi(O; \hat P, \hat d)] - \P[\varphi(O; P, d)] \\
    &= (\Pn-\P)\varphi(O; P, d) + (\Pn-\P)[\varphi(O; \hat P, \hat d)-\varphi(O; P, d)] + \P[\varphi(O; \hat P, \hat d)-\varphi(O; P, d)] , \nonumber\\
    &=S+R_1+R_2,\nonumber
\end{align}
where $S$ is the standard CLT term, $R_1$ is the empirical process term and can be easily bounded when conducting sample splitting, the bias $R_2$ can be handled by further separating errors in nuisance functionals and selectors
\begin{align*}
\P\left[\varphi(O;\hat P,\hat d)-\varphi(O;P,d)\right]
&= \underbrace{\P\left[\varphi(O;\hat P,d)-\varphi(O;P,d)\right]}_{\text{nuisance error } B_\text{nuis}}
+ \underbrace{\P\left[\varphi(O;\hat P,\hat d)-\varphi(O;\hat P,d)\right]}_{\text{selector error } B_\text{sel}}.
\end{align*}
Detailed proof is shown in \ref{proof:direct}.

\begin{remark}
The selector-rate requirement $\|\hat\theta-\theta\|_\infty = o_p(n^{-1/(2(1+\alpha))})$ is the main extra price paid by the direct method when we plug in the estimated selector is that one needs uniform error control. The parameter $\alpha$ quantifies the density of the population near the ``tie" region where $\theta_1(X) = \theta_0(X)$, representing the degree of separation between conditional potential outcome ranks. A larger $\alpha$ indicates a stronger margin with fewer ambiguous cases, which enhances the stability of the plug-in selector $\hat{d}(x)$ and yields faster convergence rates for the non-smooth direct estimator.

In applications one typically implements cross-fitting and uses modern ML estimators for \(\theta_a\) and \(\pi_a\); verifying the sup-norm condition may require specially tailored estimators (series/sieve with tuned complexity or uniformly consistent kernel estimators).
\end{remark}

\subsubsection{Smooth Approximation via Log-Sum-Exp}
In the direct method, the non-differentiability of $\phi(u,v)=\min(u,v)$
creates boundary complications. As an alternative, we approximate $\min$
by a smooth function based on the log-sum-exp operator, i.e.,
\[
g_t(u,v) := -\frac{1}{t}\log\big(e^{-t u}+e^{-t v}\big),
\quad (u,v)\in[0,1]^2,
\]
where $g_t(u,v)$ is a smooth function satisfying $\min(u,v)- \frac{\log 2}{t} \le g_t(u,v) \le \min(u,v)$ and $\lim_{t\to\infty} g_t(u,v)=\min(u,v)$. Then $\Psi(P)$ can be approximated via $\Psi_t(P)=\E\big[g_t(\theta_0(X),\theta_1(X))\big]$, which is continuously Gateaux-differentiable in $P$ for any finite $t$, and $\Psi_t(P)\uparrow\Psi(P)$ as $t\to\infty$.

\begin{theorem}[Properties of the smooth log-sum-exp estimator]\label{thm:smooth}
Let 
\[
\Psi_t=\E[g_t(\theta_0(X),\theta_1(X))],\quad g_t(u,v) = -t^{-1} \log(e^{-tu}+e^{-tv}),
\]
be the smooth approximation of $\Psi$, let $\hat\Psi_t$ be the plug-in estimator using nuisance estimators $\hat\theta_a$ and $\hat\pi_a$, $a=0,1$, obtained on an independent sample (sample splitting)
\[
\Psi_t = \Pn\varphi_t(O;\hat P)
= \Pn\!\Big[
\sum_{a=0}^1 \hat w_{a,t}(X) \frac{\1\{A=a\}}{\hat\pi_a(X)}(\1\{Y\le y\}-\hat\theta_a(X)) + g_t(\hat\theta_0(X),\hat\theta_1(X))
\Big],
\]
where
\[
\hat w_{a,t}(X)
= \frac{e^{-t\hat\theta_a(X)}}{e^{-t\hat\theta_0(X)}+e^{-t\hat\theta_1(X)}}.
\]
Suppose the boundedness assumption that there exist constants $0<\underline c<\overline c<\infty$ such that
$\underline c < \hat\pi_a(x) < \overline c$ almost surely holds. Then for a fixed $t$, the estimator $\hat\Psi_t$ satisfies the following properties
\begin{itemize}
    \item (consistency) Under rate condition on the nuisance estimators,
\[
\|\hat\theta_a-\theta_a\|_{L_2(P)}\cdot \|\hat\pi_a-\pi_a\|_{L_2(P)} = o_p(1),\quad \|\hat\theta_a-\theta_a\|^2_{L_2(P)} = o_p(1),\, a=0,1,
\] 
   $\hat\Psi_t$ is consistent for $\Psi_t(P)$.
   \item (root-$n$ consistency and asymptotic normality) under product rate condition on the nuisance estimators,
\[
\|\hat\theta_a-\theta_a\|_{L_2(P)}\cdot \|\hat\pi_a-\pi_a\|_{L_2(P)} = o_p(n^{-1/2}),\quad \|\hat\theta_a-\theta_a\|^2_{L_2(P)} = o_p(n^{-1/2}),\, a=0,1,
\] 
i.e.,
\[
    \|\hat\theta_a-\theta_a\|_{L_2(P)} = o_p(n^{-1/4}) \quad \text{and} \quad \|\hat\pi_a-\pi_a\|_{L_2(P)} = o_p(n^{-1/4})
    \]
the estimator $\hat\Psi_t$ admits the linear expansion
\[
\hat\Psi_t - \Psi_t(P) = (\Pn - P) \varphi_t(O;P) + o_p(n^{-1/2}),
\]
and therefore
\[
\sqrt{n} (\hat\Psi_t - \Psi_t(P)) \dto N(0, \Var(\varphi_t(O;P))).
\]
so that the estimator converges at the $\sqrt{n}$ rate $\sqrt{n}(\hat\Psi_t - \Psi_t(P)) = O_p(1)$.
\end{itemize}
\end{theorem}

Similarly, the asymptotic normality is established on the von-Mises decomposition
\begin{align}
\hat\Psi_t - \Psi_t(P)
&= \underbrace{(\Pn-P)\varphi_t(O;P)}_S
+ \underbrace{P[\varphi_t(O;\hat P)-\varphi_t(O;P)]}_{R_{\mathrm{nuis}}}
+ \underbrace{(\Pn-P)[\varphi_t(O;\hat P)-\varphi_t(O;P)]}_{R_n},
\end{align}
where $S$ is the standard CLT term, $R_n$ is the empirical process term, and $R_{\mathrm{nuis}}$ is the remaining bias. Detailed proof is shown in Appendix \ref{proof:smooth}.

Note that the above asymptotic properties of smooth-approximation estimator $\hat\Psi_t$ is established with a fixed smooth parameter $t$. We are curious about the performance on how the smooth estimator $\hat\Psi_t$ approximate the true parameter $\Psi(P)$, i.e., the behavior when $t$ goes to infinity. It appears that our smooth estimator will be closer and closer to the truth but we need pay price for it.

\begin{remark}[bias-variance trade-off in limit as $t\to\infty$]
To achieve the semiparametric efficiency bound, we decompose
$$\sqrt{n}(\hat\Psi_t-\Psi(P)) = \underbrace{\sqrt{n}(\hat\Psi_t-\Psi_t(P))}_{\text{Statistical Error}} + \underbrace{\sqrt{n}(\Psi_t(P)-\Psi(P))}_{\text{Approximation Bias}}$$

The term $\sqrt{n}(\Psi(P)-\Psi_t(P))$ is the smoothing bias and bounded by $O(\frac{\sqrt{n}}{t})$, and to bound it we need $\sqrt{n}/t \to 0$, i.e., $t=\omega(\sqrt{n})$.

The term $\sqrt{n}(\hat\Psi_t-\Psi_t(P))$ is the statistical error, converging to $N(0, \Var(\varphi_t(O;P)))$ for a fixed $t$. When $t$ is moving
\begin{align*} \sqrt{n}(B_{\mathrm{nuis}} + R_n) &\approx \sqrt{n} \cdot O_p(t_n \|\hat\theta-\theta\|_2^2) + \sqrt{n} \cdot O_p(n^{-1/2} t_n \|\hat\theta-\theta\|_2) \\ &= O_p(\sqrt{n} t_n \|\hat\theta-\theta\|_2^2) + O_p(t_n \|\hat\theta-\theta\|_2) \end{align*}

To make $\Psi_t(P) \uparrow \Psi(P)$ then $
\varphi_t(O;P) \to \varphi_{\mathrm{oracle}}(O;P)$ in $L_2(P)$, we need $t_n \cdot \|\hat\theta - \theta\|_{L_2(P)} = o_p(1)$ and $\sqrt{n} \cdot t_n \cdot \|\hat\theta - \theta\|_{L_2(P)}^2 = o_p(1)$, which requires $ \|\hat\theta - \theta\|_{L_2(P)} = o_p(n^{-1/2})$ rather than $ \|\hat\theta - \theta\|_{L_2(P)} = o_p(n^{-1/4})$ when $t=\omega(\sqrt{n})$. Here's an example of bias-variance trade-off.
\end{remark}

\begin{remark}[Lower bound]
    For lower bound $\phi(\theta_0, \theta_1) = \max(\theta_0(X) + \theta_1(X) - 1, 0)$, similar logic applies. We can approximate via $$g_t(\theta_0, \theta_1) = \frac{1}{t} \log\big(1 + e^{t[\theta_0(X) + \theta_1(X) - 1]}\big),$$
    and the corresponding IF-based estimator is
    $$\hat\Psi^{\text{L}}_t = \mathbb{P}_n \left[ \sum_{a=0}^1 \hat{w}_{t}(X) \frac{\1\{A=a\}}{\hat{\pi}_a(X)} \Big(\1\{Y \le y\} - \hat{\theta}_a(X)\Big) + g_t\big(\hat{\theta}_0(X), \hat{\theta}_1(X)\big) \right],$$
    where $\hat{w}_{t}(X)=\sigma(t\hat{S}) = \frac{e^{t(\hat{\theta}_0(X) + \hat{\theta}_1(X) - 1)}}{1 + e^{t(\hat{\theta}_0(X) + \hat{\theta}_1(X) - 1)}}$, $\hat{S} = \hat{\theta}_0(X) + \hat{\theta}_1(X) - 1$.
\end{remark}

Beyond addressing non-differentiability, the smooth approximation offers a crucial inferential advantage. When the true parameter is exactly on the boundary, classical normal approximations and standard bootstrap procedures are known to be inconsistent \cite{andrews2000inconsistency}, while $g_t(\cdot)$ pulls the estimand strictly into the interior of the parameter space, ensuring the asymptotic validity of Wald-type confidence intervals.

\section{Counterfactual marginals with unmeasured confounding: Triple Machine Learning via a Representation Learning based IV Approach}\label{sec:marginal-iv}
We then move beyond the simple case with no unobserved confounding to the more complex scenario where unobserved confounding is present. In this setting, even the marginal distributions of the potential outcomes are typically non-identifiable. Furthermore, even the introduction of auxiliary variables, such as IVs in two-stage least squares (2SLS) regressions, either need (strong) parametric assumption like linear structural equation models (SEMs) to identify ATE, or can only identify local average treatment effect (LATE) for compliers in a nonparametric model with additional assumptions such as monotonic relevance and binary treatments \cite{angrist1996identification,swanson2018partial}. Nonparametric IV methods have gained significant attention in recent years \cite{newey2013nonparametric,levis2024nonparametric}, but they have focused more on the identification and estimation of average treatment effects rather than the full counterfactual distribution.

The challenge of unmeasured confounding has traditionally been viewed as an identification bottleneck that can only be breached by rigid functional forms. However, the emerging field of Causal Representation Learning (CRL) \cite{scholkopf2021toward,Moran09022026} offers a alternative perspective by framing the recovery of hidden confounders as an identifiable latent variable modeling problem. Unlike standard representation learning that focuses on extracting features for optimal prediction, CRL aims to identify the underlying causal variables and their structural relations from high-dimensional observations through structural constraints given by auxiliary information like exogenous instruments. This shift from treating confounders as unreachable nuisances to treatable latent variables provides the conceptual inspiration for our framework: using exogenous instrument as the key identifying constraint to recover the latent confounding subspace.

In this section, we propose a novel method that leverages IVs and representation learning to construct a latent representation of the unmeasured confounding, followed by the non-parametric identification and semi-parametric estimation of causal functionals. We first focus on the identification and inference of the marginals or average effects of the latent structure. This approach will be conveniently integrated with the bounding methods discussed in the previous section \ref{sec:bound-cov}, which utilized conditional copulas to transition from marginal to joint distributions, thereby enabling inference on the entire joint distribution and individual-level effects, which will be discussed in Section \ref{sec:bound-iv}.

Specifically, we observe $(S, A, Y)$ with exogenous instrument ($S$), treatment ($A$), and outcome ($Y$). We suspect unobserved confounders $Z_C$ such that $A \leftarrow Z_C \to Y$. Our goal is to estimate counterfactual marginals and Average Treatment Effects (ATE). We aim to use a representation learning based approach that uses $S$ to help identify a representation $\widehat Z_C$ of the unobserved $Z_C$ and then identify the ATE (Fig.\ref{fig:rl-iv-dag}).

\subsection{Identification}
\begin{figure}[t]
\centering
\begin{tikzpicture}[>=stealth, node distance=1.8cm, thick]

\node[draw, circle, fill=gray!10, inner sep=1pt, minimum size=0.9cm] (S) {$S$};
\node[draw, circle, fill=gray!10, right of=S, node distance=2.0cm, minimum size=0.9cm] (ZS) {$Z_S$};
\node[draw, circle, fill=gray!10, below of=ZS, node distance=1.8cm, minimum size=0.9cm] (ZC) {$Z_C$};
\node[draw, circle, fill=gray!10, right of=ZS, node distance=2.0cm, minimum size=0.9cm] (A) {$A$};
\node[draw, circle, fill=gray!10, right of=A, node distance=2.0cm, minimum size=0.9cm] (Y) {$Y$};

\node[draw, rectangle, fill=blue!10, left of=ZC, node distance=1.5cm, minimum size=0.9cm] (EC) {$\varepsilon_C$};
\node[draw, rectangle, fill=blue!10, above of=ZS, node distance=1.5cm, minimum size=0.9cm] (ES) {$\varepsilon_S$};

\draw[->] (S) -- (ZS);
\draw[->] (ZC) -- (ZS);
\draw[->] (ZS) -- (A);
\draw[->] (A) -- (Y);
\draw[->] (ZC) -- (Y);

\draw[->, dashed] (EC) -- (ZC);
\draw[->, dashed] (ES) -- (ZS);

\node[above=0.0cm of S] {\scriptsize Instrument};
\node[above=0.0cm of A] {\scriptsize Treatment};
\node[above=0.0cm of Y] {\scriptsize Outcome};
\node[below=0.0cm of ZC] {\scriptsize Confounder};
\node[above=0.0cm of ZS] {\scriptsize $S$-dependent};
\node[above=0.0cm of EC] {\scriptsize Noise};
\node[above=0.0cm of ES] {\scriptsize Noise};

\end{tikzpicture}
\caption{Causal and latent structure underlying the IV-based representation-learning model with Exogenous Noises. $\varepsilon_C$ and $\varepsilon_S$ are the independent, exogenous noise sources.
$Z_C$ (pure confounder) is defined by $\varepsilon_C$. $Z_S$ (IV-related latents) is defined by $Z_C$, $S$, and $\varepsilon_S$.
The absence of direct edges $S\to Y$ and $Z_S\to Y$ satisfies the Exclusion Restriction.}
\label{fig:rl-iv-dag}
\end{figure}
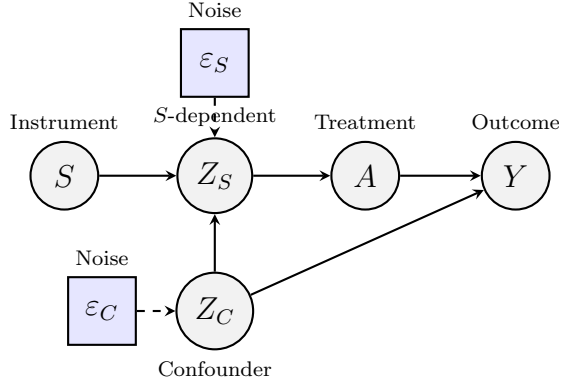

To identify the marginal distribution of potential outcomes in the presence of unobserved confounding, we exploit recent developments in identifiable CRL with exogenous auxiliary variables \cite{kong2022partial,ng2025causal,ng2025debiasing}. Specifically, we introduce an exogenous instrument $S$ to induce a decomposition of the latent representation of the observed data into components driven by the exogenous variation in $S$ ($Z_S$), and components which captures variation associated with the unobserved confounding ($Z_C$)

Since the confounders are not observed, $Z_C$ and $Z_S$ cannot be simply assumed to be conditionally independent. Nevertheless, under a sufficiently variability condition on $p(A|S)$, $Z_C$ remains identifiable. In practice, since both $S$ and the treatment $A$ are observed, it is unnecessary to estimate $Z_S$ explicitly. Instead, we adopt a reduced-form construction that directly targets $Z_C$, thereby alleviating the additional randomness that would arise from simultaneously estimating multiple latent components (Section \ref{subsec:learn}).

Note that, given the dimensionality of the observed variables (which constraints the latent dimensionality), $Z_C$ should not be interpreted as a direct representation of the high-dimensional confounders themselves. Rather, it serves as a low-dimensional sufficient score summarizing the aggregate influence of the unobserved confounding. We have the following identifiability results.

\begin{theorem}[Identification]\label{thm:rep-iv-ate}
Let the observed data be i.i.d.\ samples of $(S,A,Y)$, where
$S$ is an instrument, $A$ is the treatment, and $Y$ is the observed outcome. Assume the data-generating process admits latent variables $Z=(Z_C,Z_S)$ and exogenous noises $\varepsilon_C,\varepsilon_S$ satisfying the following assumptions:

\begin{assumption}[Causal structural equations]\label{assump:A1}
The latent variables are generated as $Z_C=f(\varepsilon_C)$ and $Z_S=h(Z_C,S,\varepsilon_S)$.
The observed variables $(A,Y)$ are generated via
\[
A=g_A(Z_S),\qquad Y=g_Y(A,Z_C),
\]
where
\begin{itemize}
    \item $g_A$ is an arbitrary measurable function (allowing for discrete/binary treatments).
    \item $g_Y(a, \cdot)$ is injective with respect to $Z_C$ for each observed treatment level $a \in \mathcal{A}$.
\end{itemize}
Consequently, the joint mapping $(Z_C, Z_S) \mapsto (A, Y)$ allows the unique recovery of $Z_C$ from observables $(A,Y)$ (i.e., $Z_C$ is observable-measurable). All exogenous noises are mutually independent and $Z_C\perp\!\!\!\perp S$. 
\end{assumption}

\begin{assumption}[IV conditions]\label{assump:A2}
Variable $S$ is a valid instrument satisfying
\begin{itemize}
  \item Exclusion Restriction: $S$ does not directly affect $Y$, i.e., 
        $Y = g_Y(A, Z_C)$ does not depend on $S$, or in the notations of potential outcomes $Y(s,a)=Y(a)$.
  \item Unconfoundedness: $S\perp\!\!\!\perp (\varepsilon_C,\varepsilon_S)$, or equivalently $S\perp\!\!\!\perp (A(s),Y(s))$.
\end{itemize}
\end{assumption}

\begin{assumption}[Sufficient Variability of IV]\label{assump:variability}
The instrument $S$ induces sufficient variation in the latent variable $Z_S$ conditional on $Z_C$. Formally, for almost all $z_c$, the family of conditional distributions $\mathcal{P}_{z_c} = \{ P(Z_S \in \cdot \mid Z_C=z_c, S=s) : s \in \mathcal{S} \}$ is complete.
In terms of sets, this implies that for any measurable set $\mathcal{B} \subseteq \mathcal{Z}_S$ with positive measure ($P(\mathcal{B}) > 0$), the probability mass assigned to this set varies with $S$
\[
s \mapsto P(Z_S \in \mathcal{B} \mid Z_C=z_c, S=s) \quad \text{is not a constant function a.s.}
\]
\end{assumption}

\begin{assumption}[Positivity and regularity]\label{assump:A3}
For almost all $z_C$ in the support of $Z_C$, $0< P(A=a\mid Z_C=z_C)<1$, and required conditional expectations exist and are finite.
\end{assumption}

\medskip

Then the confounding subspace $Z_C$ is \textbf{subspace-identifiable}. In other words, if there exists a learned representation $\widehat{Z}_C$ that satisfies independence constraints $\widehat{Z}_C \perp\!\!\!\perp S$ and predictive sufficiency $Y \perp\!\!\!\perp Z_C \mid (A, \widehat{Z}_C)$ or $p(Y \mid A, \widehat{Z}_C, Z_C) = p(Y \mid A, \widehat{Z}_C)$ almost surely, then there exists an invertible map $\psi$ such that $\widehat Z_C = \psi(Z_C)$ almost surely. 

And subsequently, the counterfactual marginals and the ATE is nonparametrically identifiable
\[
\mathbb E[Y(a)] = \mathbb E\big[\,\mathbb E[Y\mid A=a, \widehat Z_C]\,\big],\quad
\mathrm{ATE}(a,a') = \mathbb E\big[\,\mathbb E[Y\mid A=a, \widehat Z_C] - \mathbb E[Y\mid A=a', \widehat Z_C]\,\big].
\]
\end{theorem}

\begin{remark}[Completeness v.s. Variability]
Assumption \ref{assump:variability} is formally equivalent to the \emph{Bounded Completeness} condition often used in econometrics. We adopt the terminology of \emph{Sufficient Variability} \citep{kong2022partial, ng2025debiasing} to highlight the geometric intuition: the instrument $S$ must create diverse distributional shifts in $Z_S$ such that no subset of the latent space remains ``invariant" to $S$. This property is crucial in Step 1 of the proof to rule out representations that entangle $Z_S$.
\end{remark}

\begin{remark}[On the injectivity Assumption]
We emphasize that the latent variables $Z_C$ and $Z_S$ should not be interpreted as the high-dimensional raw state of the world (e.g., all genetic factors). Instead, following the philosophy of Sufficient Dimension Reduction or Control Functions, they represent the low-dimensional (or scalar) projections or scores of these factors (so that we can let $\dim(Z_C) \le \dim(Y)$) that actively drive the variation in $A$ and $Y$. Under this interpretation, the injectivity assumption in \ref{assump:A1} implies that the causal mechanism relies on a low-dimensional bottleneck, which is a standard structural assumption in representation learning \cite{kong2022partial,ng2025causal,ng2025debiasing}. 
\end{remark}

\begin{remark}[Identifiability of counterfactual marginals]
Theorem \ref{thm:rep-iv-ate} implies the identifiability of marginal distribution of potential outcome $Y(a)$, i.e.
\[
P(Y(a)\leq y)=\E\left[P(Y\leq y|A=a,\hat Z_c)\right].
\]
\end{remark}

The proof of Theorem \ref{thm:rep-iv-ate} consists of two main stages. First, based on the injectivity and completeness assumptions, we establish the identifiability of the confounding subspace $Z_C$ up to an invertible transformation (Step 1). Second, we utilize this identified representation to identify the marginal counterfactual distributions $F_{Y(a)}(y)$ and subsequently, the Average Treatment Effect (ATE) via the back-door adjustment formula (Steps 2-4). Detailed proof is shown in Appendix \ref{proof:rep-iv-ate}.

\subsection{Estimation and Inference}
Generally speaking, we can estimate the average potential outcome $\psi(a) = \mathbb{E}[Y(a)]$ using similar approaches but replace the observed confounding $Z_C$ to the confounding representations $\hat{Z}_C$ learned from a separated dataset. We quickly review the standard ways as follows.

\subsubsection{Standard estimators with observed confounding}
For binary treatment $A \in \{0, 1\}$, the outcome regression (OR) estimator is $\hat{\psi}_{\text{OR}}(a) = \frac{1}{n}\sum_{i=1}^n \hat{\mu}(a, Z_{C,i})$, where $\hat{\mu}(a, z) = \mathbb{E}[Y | A=a, Z_C=z]$ is the estimated outcome model. The inverse propensity weighting (IPW) estimator is $\hat{\psi}_{\text{IPW}}(a) = \frac{1}{n}\sum_{i=1}^n \frac{\1(A_i=a)}{\hat{\pi}(a|Z_{C,i})} Y_i$, where $\hat{\pi}(a|z) = P(A=a|Z_C=z)$ is the propensity score. The doubly robust (DR) estimator combines both models as $\hat{\psi}_{\text{DR}}(a) = \frac{1}{n}\sum_{i=1}^n \left[\hat{\mu}(a, Z_i) + \frac{\1(A_i=a)}{\hat{\pi}(a|Z_{C,i})} (Y_i - \hat{\mu}(A_i, Z_{C,i}))\right]$ and admits the double robustness \citep{bang2005doubly}.

For continuous treatment $A \in \mathbb{R}$, the indicator function $\1(A_i=a)$ has zero probability, requiring density or kernel-based adaptations. Let $\hat{r}(a|z) = \hat{f}(a|z)$ denote the estimated conditional density (generalized propensity score, GPS) and $K_h(\cdot)$ be a kernel function with bandwidth $h$. The outcome regression estimator remains $\hat{\psi}_{\text{OR}}(a) = \frac{1}{n}\sum_{i=1}^n \hat{\mu}(a, Z_{C,i})$. The GPS-IPW estimator \citep{hirano2004propensity} is $\hat{\psi}_{\text{GPS}}(a) = \frac{\sum_{i=1}^n K_h(A_i - a) Y_i / \hat{r}(A_i|Z_{C,i})}{\sum_{i=1}^n K_h(A_i - a) / \hat{r}(A_i|Z_{C,i})}$, where $K_h(A_i - a)$ weights observations near $a$ and $1/\hat{r}(A_i|Z_{C,i})$ provides inverse density weighting. The density estimation-based standard DR estimator \citep{kennedy2017non} can be constructed as $\hat{\psi}_{\text{D}}(a) = \frac{1}{n}\sum_{i=1}^n \left[\hat{\mu}(a, Z_{C,i}) + \frac{\hat{r}(a|Z_{C,i})}{\hat{r}(A_i|Z_{C,i})} (Y_i - \hat{\mu}(A_i, Z_{C,i}))\right]$, where the density ratio $\hat{r}(a|Z_{C,i})/\hat{r}(A_i|Z_{C,i})$ replaces the indicator-based weight from binary DR, maintaining doubly robust consistency. The DR-kernel estimator tends to employ kernel methods and simplifies this as $\hat{\psi}_{\text{K}}(a) = \frac{1}{n}\sum_{i=1}^n \left[\hat{\mu}(a, Z_{C,i}) + n \cdot w_i(a) (Y_i - \hat{\mu}(A_i, Z_{C,i}))\right]$, where $w_i(a) = K_h(A_i - a) / \sum_{j=1}^n K_h(A_j - a)$ are normalized kernel weights that avoid explicit density estimation, making it computationally efficient while maintaining approximate doubly robust properties \citep{flores2012estimating, fong2018covariate}.

The GPS can be estimated using kernel density estimation (KDE) on residuals $\hat{r}(a|Z) \approx \frac{1}{n}\sum_{j=1}^n K_h((a - \hat{\mu}(Z)) - r_j)$, where $r_j = A_j - \hat{\mu}(Z_j)$ are training residuals and $\hat{\mu}(Z) = \mathbb{E}[A|Z]$ is the conditional mean of treatment estimated via gradient boosting. This nonparametric approach accommodates non-Gaussian treatment distributions. For numerical stability, extreme weights can be trimed. Bandwidth selection follows Silverman's rule $h = 1.06 \hat{\sigma}_A n^{-1/5}$ \cite{silverman2018density}, where $\hat{\sigma}_A$ is the sample standard deviation of treatment.

\subsubsection{Triple Machine Learning estimator}\label{subsec:triple-est}
We now resort to learning confounding representations when there are unmeasured confounders, and subsequently, we can employ all the above standard estimators by replacing the observed covariates with learned confounding representations. It mimic the standard semiparametric/double-robust/double-ML way \cite{chernozhukov2018double,kennedy2024semiparametric}, and we call it ``Triple Machine Learning" (TML) since we need an additional fold to learn the confounding representation first before the standard double machine learning or double robust estimation. As an example, for binary treatments, we can construct an estimator $\widehat{\E}[Y(a)]$ using three-fold sample splitting or more complex cross-fitting. We summarize the procedure as follows
\begin{itemize}
        \item Fold 1 to learn representation $\widehat{Z}_C = \widehat{\psi}(A,Y)$.
        \item Fold 2 to estimate nuisance parameters
        \begin{itemize}
            \item $\widehat{m}_a(z) = \hat\E[Y \mid A=a, \hat Z_C=z]$ in the Outcome model.
            \item $\widehat{\pi}_a(z) = \hat P(A=a \mid \hat Z_C=z)$ in the Treatment (Propensity) model.
        \end{itemize}
        \item Fold 3 to estimate the causal parameter
        \[
        \widehat{\E}[Y(a)]=\Pn\Big[\widehat{m}_a(\widehat{Z}_{C}) + \frac{\1\{A=a\}}{\widehat{\pi}_a(\widehat{Z}_{C})} \big(Y_i - \widehat{m}_a(\widehat{Z}_{C})\big)
        \Big]
        \]
    \end{itemize}
It can easily generalize to other distributional targets by modifying the outcome model to other functionals like $\widehat{m}_a(z, y) = \widehat{P}(Y \leq y \mid A=a, \widehat{Z}_C=z)$. Moreover, similar logic can be applied to the continuous scenario with GPS, density or kernel based estimators. In principle, the proposed TML framework is sufficiently general to accommodate a variety of estimators, including outcome regression, inverse probability weighting, and doubly robust estimators. We first establish the theoretical properties of the TML estimator under binary treatment assignment using a doubly robust construction (\ref{thm:dr-rliv-summary}), and it can be easily generalized to more complex continuous scenarios. We subsequently introduce a variational approach for effective confounding representation learning (Section \ref{subsec:learn}), and assess the numerical performance of different estimators within the TML framework through simulation studies (Section \ref{sec:sim}). 
    
\begin{theorem}[Asymptotic Properties of the Doubly-Robust-Type Representation-IV Estimator]
\label{thm:dr-rliv-summary}
Let the identification assumptions of Theorem~\ref{thm:rep-iv-ate} hold.
Let the estimator $\widehat{\E}[Y(a)]$ be constructed using $K$-fold cross-fitting, where $\widehat{Z}_C = \widehat{\psi}(Z_C)$, $\widehat{m}$, and $\widehat{\pi}$ are estimated on separate folds. Let the boundedness and overlap assumption hold, i.e., there exist constants $0<\underline c<\overline c<\infty$ such that
$\underline c < \pi(a\mid z) < \overline c$ almost surely,
and $\mathbb{E}[Y^2]<\infty$.

Then we have
\paragraph{1. Consistency (double robustness)}
The estimator $\widehat{\E}[Y(a)]$ is consistent for $\E[Y(a)]$, i.e.,
\[
\widehat{\E}[Y(a)] \xrightarrow{p} \E[Y(a)],
\]
if the following conditions hold
\begin{enumerate}[label=(C1\alph*)]
    \item The representation $\widehat{Z}_C$ is consistent
    \[
    \|\widehat{Z}_C - \psi(Z_C)\|_{L_2,P} = o_p(1).
    \]
    \item At least one of the nuisance estimators is consistent for the true function defined on $Z_C$, despite being trained on $\widehat{Z}_C$ (i.e., it is EIV-robust and consistent)
    \[
    \|\widehat{m}(\hat Z_C) - m_0(Z_C)\|_{L_2,P} = o_p(1) \quad \text{or} \quad \|\widehat{\pi}(\hat Z_C) - \pi_0(Z_C)\|_{L_2,P} = o_p(1).
    \]
\end{enumerate}

\paragraph{2. Asymptotic Normality}
The estimator is $\sqrt{n}$-asymptotically normal with a corrected variance,
\[
\sqrt{n}\,\bigl(\widehat{\E}[Y(a)] - \E[Y(a)]\bigr)
\xrightarrow{d} \N(0, V_\text{total}(a)),
\]
where $V_\mathrm{total}(a) = \mathrm{Var}[\mathrm{IF}_a(O; \eta_0) + \mathrm{IF}_{\phi, \mathrm{rep}}(a)]$, if the following hold
\begin{enumerate}[label=(C2\alph*)]
    \item The nuisance functions (assumed EIV-robust) satisfy the double machine learning (DML) rate condition
    \[
    \|\widehat{m}(\hat Z_C) - m_0(Z_C)\|_{L_2,P} \cdot \|\widehat{\pi}(\hat Z_C) - \pi_0(Z_C)\|_{L_2,P} = o_p(n^{-1/2}).
    \]
    \item The representation $\widehat{Z}_C$ converges as
    \[
    \|\widehat{Z}_C - \psi(Z_C)\|_{L_2,P} = O_p(n^{-1/2}).
    \]
\end{enumerate}
In this case, the first-stage estimation error in $\widehat{Z}_C$ contributes a first-order correction term $\mathrm{IF}_{\phi, \mathrm{rep}}(a)$ to the influence function.

\paragraph{3. Efficiency}
The estimator is $\sqrt{n}$-asymptotically normal and achieves the ordinary efficient variance lower bound,
\[
\sqrt{n}\,\bigl(\widehat{\E}[Y(a)] - \E[Y(a)]\bigr)
\xrightarrow{d} \N(0, V_{a}),
\]
where $V_a = \mathrm{Var}[\mathrm{IF}_a(O; \eta_0)]$, if the following hold
\begin{enumerate}[label=(C3\alph*)]
    \item The nuisance functions (assumed EIV-robust) satisfy the DML rate condition
    \[
    \|\widehat{m}(\hat Z_C) - m_0(Z_C)\|_{L_2,P} \cdot \|\widehat{\pi}(\hat Z_C) - \pi_0(Z_C)\|_{L_2,P} = o_p(n^{-1/2}).
    \]
    \item The representation $\widehat{Z}_C$ converges at a ``super-convergence" rate
    \[
    \|\widehat{Z}_C - \psi(Z_C)\|_{L_2,P} = o_p(n^{-1/2}).
    \]
\end{enumerate}
In this case, the first-stage estimation error is asymptotically negligible, and the correction term $\mathrm{IF}_{\phi, \mathrm{rep}}(a)$ disappears from the asymptotic expansion.

\end{theorem}

The proof relies on the von-Mises decomposition
\begin{align}
\sqrt{n}(\widehat{\psi} - \psi_0) &= \sqrt{n}(\mathbb{P}_{n,3} \varphi(\widehat{\eta}) - P \varphi(\eta_0)) \nonumber \\
&= \underbrace{\sqrt{n}(\mathbb{P}_{n,3} - P) \varphi(\eta_0)}_{T_1: \text{Oracle CLT}} 
+ \underbrace{\sqrt{n}(\mathbb{P}_{n,3} - P) (\varphi(\widehat{\eta}) - \varphi(\eta_0))}_{T_2: \text{Empirical Process}} 
+ \underbrace{\sqrt{n} P (\varphi(\widehat{\eta}) - \varphi(\eta_0))}_{T_3: \text{Bias Term}},
\end{align}
where the bias consists of a nuisance estimation error and a representation learning error
\[
T_3 = \underbrace{\sqrt{n} P (\varphi(\widehat{\eta}) - \varphi(\tilde{\eta}))}_{T_{3a} \text{ (Nuisance Bias)}} + \underbrace{\sqrt{n} P (\varphi(\tilde{\eta}) - \varphi(\eta_0))}_{T_{3b} \text{ (Representation Bias)}},
\]
where the intermediate parameter $\tilde{\eta} = (m_0, \pi_0, \widehat{\phi})$ represents the ideal nuisance parameters given the learned representation $\widehat{\phi}$. Note that $T_{3a}$ can be handled via ordinary Neyman orthogonality technique while $T_{3b}$ fails since $\psi$ is NOT orthogonal w.r.t. $Z_C$.  Hence, let $M(\phi) = \E[\varphi(O; m_0, \pi_0, \phi)]$ be the expected score functional, then we have
\[
T_{3b} = \sqrt{n} (M(\widehat{\phi}) - M(\phi_0)) = \underbrace{\sqrt{n} \nabla_{\phi} M(\phi_0)[\widehat{\phi} - \phi_0]}_{\text{Linear Term (I)}} + \underbrace{\sqrt{n} \mathcal{R}(\widehat{\phi}, \phi_0)}_{\text{Remainder Term (II)}},
\]
where the variance of the linear term is $\Var(\mathrm{IF}_{\phi, \text{rep}}(O))$. In other words, the representation error accumulates in the subsequent causal estimation, resulting in inflated variance induce by $\mathrm{IF}_{\phi, \text{rep}}(O) = \mathbb{E}_{O'}[\nabla_\phi \varphi(O'; \eta_0)] \cdot \xi_\phi(O)$ where $\xi_\phi(O)$ is the influence function for the representation learner, even when the representation has a parametric convergence rate. Detailed proof is shown in Appendix \ref{proof:dr-rliv-summary}.

In practice, the explicit calculation of the variance inflation term $V_{\text{rep}} = \Var(\mathrm{IF}_{\phi, \text{rep}})$ presents a significant challenge, as the correction influence function $\mathrm{IF}_{\phi, \text{rep}}$ depends on the infinite-dimensional sensitivity of the causal functional and the algorithmic response of the representation learner. Several empirical methods can be applied to facilitate robust inference. First, one may employ a Hessian-based numerical approximation leveraging the recent advances in neural influence functions \cite{koh2017understanding}. By treating the encoder as a parametric model $\phi_\theta$, the term $\xi_\phi(O)$ can be approximated using the inverse Hessian-vector product (HVP). Specifically, $\widehat{\mathrm{IF}}_{\phi, \text{rep}} \approx \mathbb{P}_n[\nabla_\theta \varphi]^\top H_\theta^{-1} \nabla_\theta L$, where the inverse Hessian $H_\theta^{-1}$ is efficiently computed via stochastic estimation (e.g., the LiSSA algorithm). This approach provides a first-order approximation of how small perturbations in the representation training set propagate to the final causal estimate.

Alternatively, a more computationally intensive but non-parametric approach is ensemble-based uncertainty quantification with Bootstrap or Infinitesimal Jackknife (IJ) by repeating procedure over $B$ random sample splits \cite{wager2018estimation}, and the excess variance observed across different representation learning folds effectively captures the stochasticity encoded in $V_{\text{rep}}$. While computationally demanding, this ensemble approach bypasses the need for second-order derivatives and is inherently more robust to the non-convex landscape of deep neural networks.

\begin{remark}[The EIV-Robustness Assumption]
The \textbf{EIV-robustness} underpinning all three scenarios requires the existence of estimators for $\widehat{m}$ and $\widehat{\pi}$ that can be trained on the error-laden proxy $\widehat{Z}_C$ instead of true $Z_C$ and yet still converge to the true functions $m_0$ and $\pi_0$ at the specified rates. This is a non-trivial requirement as standard estimators in non-parametric EIV settings often suffer from severe attenuation bias and degraded logarithmic convergence rates \cite{fan1993nonparametric}. Satisfying this condition implicitly requires the underlying true functional space to be sufficiently smooth, or necessitates the deployment of specialized deconvolutional algorithms during the nuisance training phase to actively correct for the measurement error in $\widehat{Z}_C$.
\end{remark}

We summarize the observations as follows.
\begin{remark}
    Within TML framework, A DR-type estimator can be constructed, but its asymptotic properties critically depend on two factors
    \begin{enumerate}
        \item \textbf{The EIV-Robustness Assumption:} We need $\widehat{m}, \widehat{\pi}$ to be trained consistently on the estimated representation $\widehat{Z}_C$ instead of the oracle one.
        \item \textbf{The Convergence Rate of $\widehat{Z}_C$}: under the EIV-Robustness Assumption,
        \begin{itemize}
            \item if $\widehat{Z}_C$ admits a $O_p(n^{-1/2})$ rate $\implies$ the causal estimator has an inflated variance $V_{total}$.
            \item if $\widehat{Z}_C$ admits a super smooth $o_p(n^{-1/2})$ rate $\implies$ the causal estimator achieve the semiparametric efficient variance $V_{a}$. It is generally very hard for a nonparametric flexible estimators like ML methods.
        \end{itemize}
    \end{enumerate}
\end{remark}

\subsection{IV-VAE: Learning representation $Z_C$}\label{subsec:learn}
The learning of $Z_C$ is a classic problem in disentangled representation learning. A natural way is to consider the following regularized $\beta$-VAE model
\begin{align}
    \mathcal{L}=&-\E_{q(Z|A,Y)}\Big[\log p(Y|Z_C,A)+\log p(A|Z_S)\Big]\nonumber+\beta KL\Big(q(Z|A,Y)||p(Z_S|S,Z_C)p(Z_C)\Big)\nonumber\\
    &+\lambda\E_q\Big[\text{HSIC}(Z_C,S)+\text{HSIC}(Z_C,A|Z_S)\Big],
\end{align}
where the (conditional) dependence can be measured via (conditional) HSIC \cite{fukumizu2004dimensionality,fukumizu2007kernel,sheng2023distance}. This involves the simultaneous optimization of two coupled objects, $\hat Z_C$ and $\hat Z_S$ (since there exists a link $Z_C\to Z_S$), and relies on the intrinsically difficult measurement and constraint of conditional independence \cite{shah2020hardness,hehardness}, making it tricky to achieve stable learning.

However, in fact, explicitly inferring $Z_S$ is computationally redundant for the purpose of confounder identification, given the fact that the instruments $S$ is observed and exogenous. We therefore employ a reduced-form specification for the treatment decoder. By substituting the mechanism of $Z_S$ into the treatment assignment function, we approximate the composite function $A = g(h(S), Z_C) = \tilde{g}(S, Z_C)$. Consequently, we can introduce the IV-VAE, where the generative model (Decoder) is defined as
$$\begin{aligned}
p_\theta(A | S, Z_C) &= \mathcal{N}(\mu_A(S, Z_C), \sigma^2_A) \\
p_\theta(Y | A, Z_C) &= \mathcal{N}(\mu_Y(A, Z_C), \sigma^2_Y)
\end{aligned}$$
Conditioning directly on $S$ allows the decoder to capture the variation in $A$ induced by the instrument pathway without the need to explicitly model the intermediate latent variable $Z_S$, while $Z_C$ captures the remaining confounding variation.

We introduce an inference network (VAE Encoder) $q_\phi(Z_C | A, Y, S)$, which conditions on $S$ to facilitate the separation of instrument-induced variation from confounder-induced variation. To enforce the structural assumption that the recovered confounder is exogenous to the instrument, we utilize the HSIC \cite{gretton2005measuring} as a regularization term to explicitly penalize statistical dependence between the learned latent space $\hat{Z}_C$ and the instrument $S$, ensuring that $\hat{Z}_C$ satisfies the properties of a valid confounder. The final optimization goal becomes
\begin{align}
    \mathcal{L}(\theta, \phi) &= \underbrace{-\mathbb{E}_{q_\phi}[\log p_\theta(A | S, Z_C) + \log p_\theta(Y | A, Z_C)]}_{\text{Reconstruction Loss}} + \beta \underbrace{{KL}(q_\phi(Z_C | \cdot) \| p(Z_C))}_{\text{KL Divergence}} \nonumber\\
    &\quad+ \lambda \underbrace{\text{HSIC}(Z_C, S)}_{\text{Independence Constraint}}.
\end{align}
This new objective focuses solely on the learning of $\hat Z_C$ and only necessitates constraining unconditional independence, which significantly facilitates stable, efficient learning and optimization.

\section{From ATE to ITE: Bounds of counterfactual joint distribution in the presence of unobserved confounding}\label{sec:bound-iv}
Identification Theorem \ref{thm:rep-iv-ate} guarantees the identifiability of conditional and marginal distributions of potential outcomes $F_{Y(a)|\hat Z_C}(y)$ and $F_{Y(a)}(y)$. We can then follow results in section \ref{sec:bound-cov} to use (conditional) copulas and FH inequality to further bound their (conditional) joint distributions. 

We now work with estimated representation $\hat Z_C$ instead of observing $X$, and hence we need to extra price on the representation learning error. Therefore, the estimation quality will be jointly affected by 1) Convergence rate of learning $\hat Z_C$; 2) EIV robustness in estimating nuisance parameters $\hat m$ and $\hat\theta$; 3) The realization of the margin condition for direct estimator or the bias-variance trade-off for the smooth approximation with different smooth parameter $t$. Cross-fitting will be generally helpful to separate these error of different sources.

\section{Simulations}\label{sec:sim}
In the following sections, we conducted empirical evaluations on both simulated and real-world data to further highlight our proposed methods. We first evaluated the conditional rank-preserving FH upper bounds of the direct estimator and the log-sum-exp approximation under varying smoothing parameters. Subsequently, we demonstrated the accuracy of confounding representation learning in the presence of unobserved confounders, along with the estimation performance of ATE and continuous treatment response curves for different estimators within the TML framework.

\subsection{Simulations on boundary estimation of counterfactual joint distribution}
To validate the theoretical properties and estimation performance of the proposed bounds under observed confounding, we conducted a simulation study with $N=2,000$ samples over 100 replications. The covariates were generated as $\mathbf{X} \sim \mathcal{N}(0, I_2)$ and treatment assignment followed a logistic model $P(A=1|\mathbf{X}) = \sigma(0.5 X_1 - 0.3 X_2 + \epsilon_S), \quad \epsilon_S \sim \mathcal{N}(0, 0.1^2)$ where $\sigma(\cdot)$ denotes the sigmoid function. 
We designed two different data generating processes (DGPs): a Linear SCM where 
$$\begin{aligned}
Y(0) &= X_1 + 0.5 X_2 + \epsilon_Y \\
Y(1) &= Y(0) + 1.0,
\end{aligned}$$
and a Nonlinear setting
$$\begin{aligned}
Y(0) &= \sin(2 X_1) + X_2^2 + \epsilon_Y \\
Y(1) &= \cos(2 X_1) + X_2^2 + 0.5 + \epsilon_Y,
\end{aligned}$$
where $\epsilon_Y \sim \mathcal{N}(0, 1)$.

First, to demonstrate the gain in identification power (tightness), we compared the width of the standard marginal bounds (Makarov bounds) against the proposed conditional bounds calculated using the true oracle distributions. The width reduction is defined as the difference between the marginal width $\mathcal{W}_{\text{marg}} = \min(F_1, F_0) - \max(F_1+F_0-1, 0)$ and the expected conditional width $\mathcal{W}_{\text{cond}} = \mathbb{E}_X[\min(F_{1|X}, F_{0|X}) - \max(F_{1|X}+F_{0|X}-1, 0)]$. We can observe that covariates-assisted conditional copulas lead to a significant tightening of bounds (Fig.\ref{fig:bound}a).

Second, we evaluated the finite-sample performance of the proposed estimators. We implemented the Doubly Robust (DR-Direct) estimator and the DR-Smooth estimator (with varying smoothing parameters) using 5-fold cross-fitting to minimize overfitting. We assessed Bias and Mean Squared Error (MSE$=$bias$^2+$SE$^2$) in each replicate (Fig.\ref{fig:bound}b). Simulation results demonstrate the clear superiority of the smooth approximation over the direct estimator, provided the smoothing parameter $t$ is sufficiently large. 

\begin{figure}[!htbp]
    \centering
  \begin{subfigure}{0.35\textwidth}
    \centering
    \includegraphics[width=\linewidth]{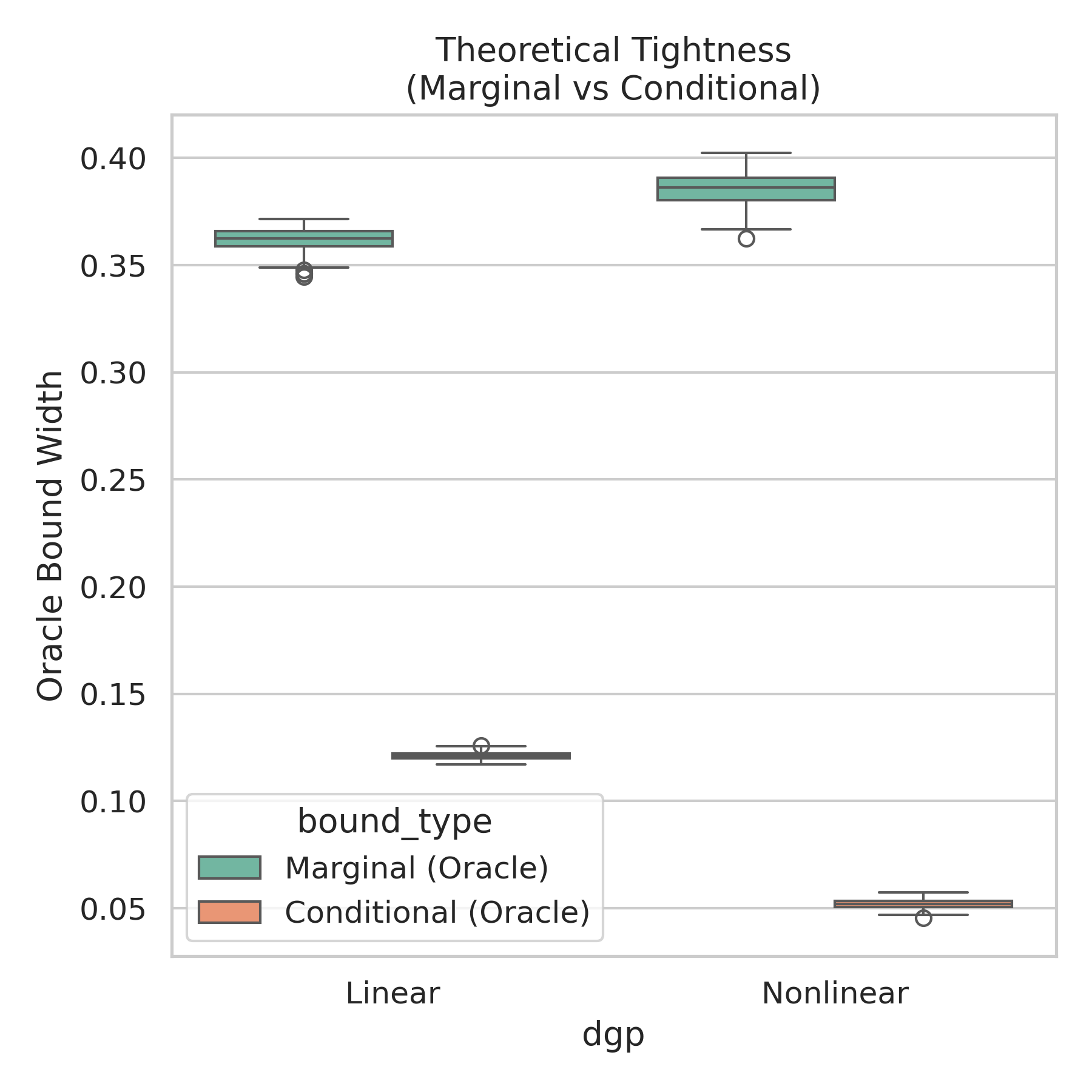}
    \caption{Bound width gained via marginal copulas and conditional copulas.}
  \end{subfigure}
  \hfill
  \begin{subfigure}{0.6\textwidth}
    \centering
    \includegraphics[width=\linewidth]{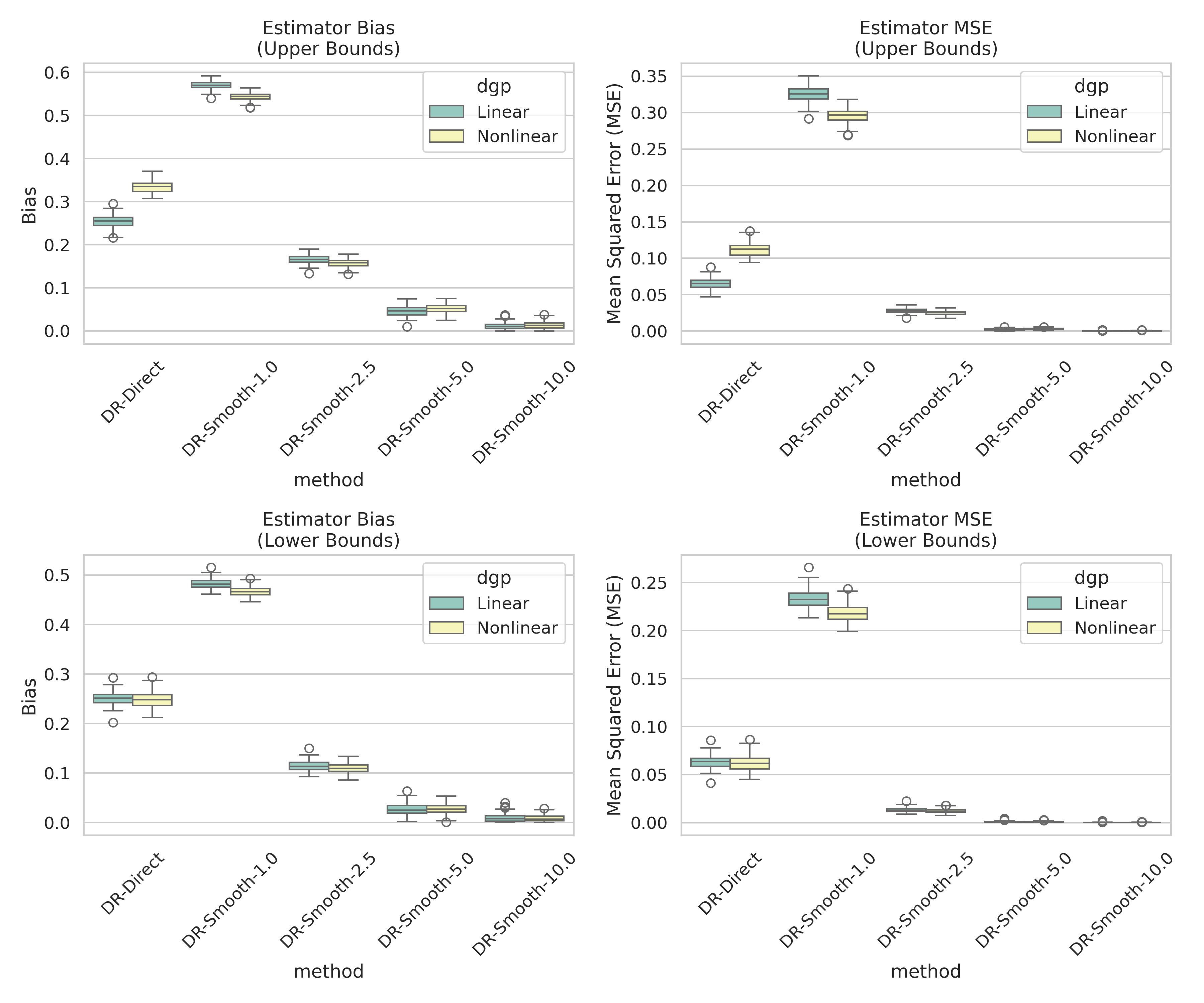}
    \caption{Boxplots of bias and MSE of estimators in each replicate.}
  \end{subfigure}

  \caption{Estimations on the bounds of joint distribution of potential outcomes.}
  \label{fig:bound}
\end{figure}

\subsection{Simulations on marginal estimation with triple matching learning approach}
To evaluate the efficacy of the proposed representation learning-based IV estimator, we conducted simulations iwith unobserved confounding. The DGP introduces a valid instrument $S \sim \mathcal{U}(-2, 2)$ and a latent confounder $Z_C \sim \mathcal{N}(0, 1)$, ensuring the crucial independence assumption $S \perp Z_C$. The treatment is generated with sufficient variability and nonlinear dependence on the instrument
$$Z_S = 0.5 S + 0.1 \tanh(S) + 0.3 \sigma(2 S) + 0.2 Z_C + \epsilon_S$$
where $\sigma(\cdot)$ is the sigmoid function and $\epsilon_S \sim \mathcal{N}(0, 0.2^2)$. 
We considered two treatment regimes: a binary treatment where $A \sim \text{Bernoulli}(\sigma(Z_S))$, and a continuous treatment where $A = Z_S$. The outcome $Y$ was generated under two distinct scenarios: In the linear setting, the outcome is a simple additive function:
$$Y_{\text{lin}} = 1 + 2A + 3Z_C + \epsilon_Y$$
In the nonlinear setting, the outcome includes sinusoidal effects, sigmoid transformations, and treatment-confounder interactions ($A \times Z_C$) to violate standard linear IV assumptions:
$$Y_{\text{nonlin}} = 1 + \left(0.3A + 0.2\sin(2A + 0.5)\right) + \left(0.3Z_C + 2\sigma(Z_C)\right) + 0.2 A Z_C + \epsilon_Y$$
We simulated $N=6000$ samples and employed a triple-split cross-fitting strategy ($K=6$ folds) to separate representation learning (VAE training), nuisance parameter estimation, and final estimation (can apply different estimators). 

\subsubsection{Representation learning quality}
We first demonstrated the quality of latent confounding representation learning. We illustrated the correlation between true $Z_C$ and learned $\hat Z_C$, and the dependence between learned $\hat Z_C$ and instrument $S$ (Fig.\ref{fig:rep}). As $Z_C$ is identifiable only up to an invertible transformation, both strong positive and strong negative correlations between $Z_C$ and $\hat Z_C$ provide evidence of successful representation learning.

\begin{figure}[!htbp]
    \centering
  \begin{subfigure}{0.49\textwidth}
    \centering
    \includegraphics[width=\linewidth]{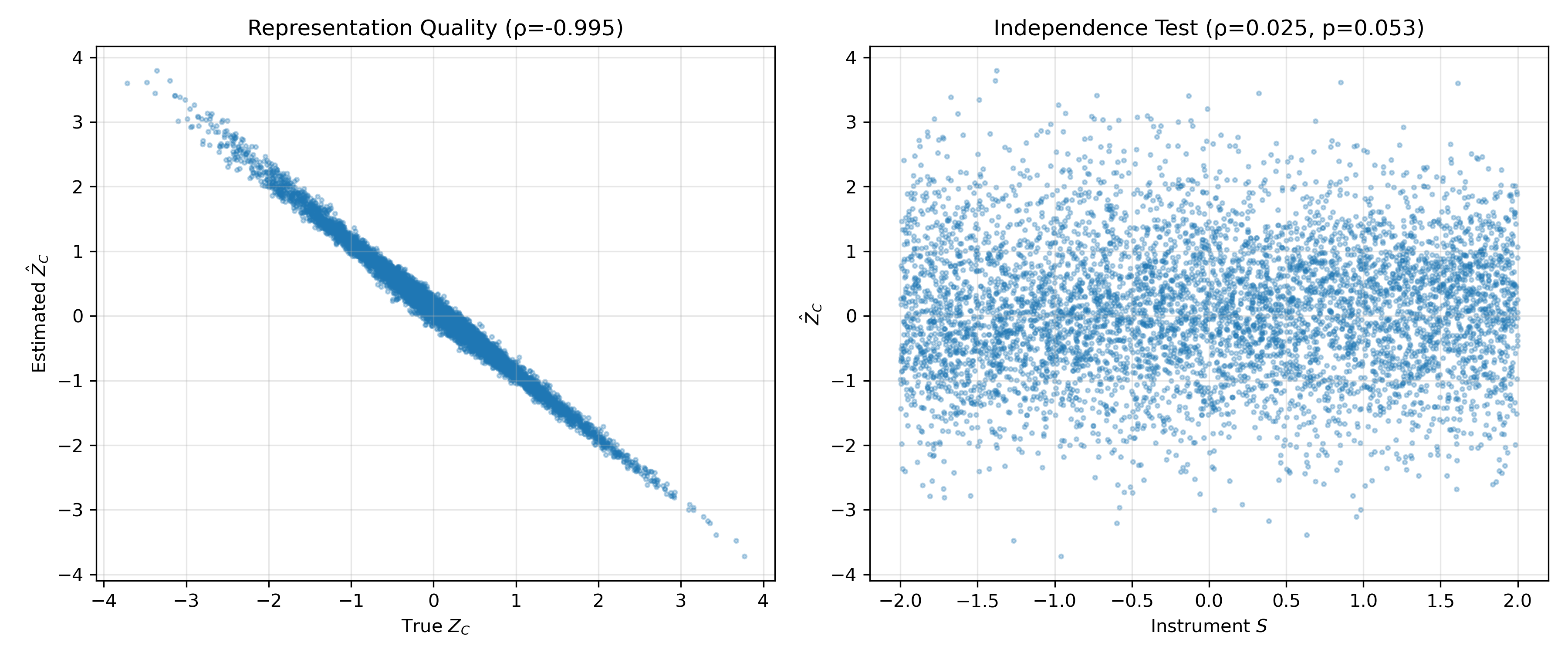}
    \caption{Linear DGP.}
  \end{subfigure}
  \hfill
  \begin{subfigure}{0.49\textwidth}
    \centering
    \includegraphics[width=\linewidth]{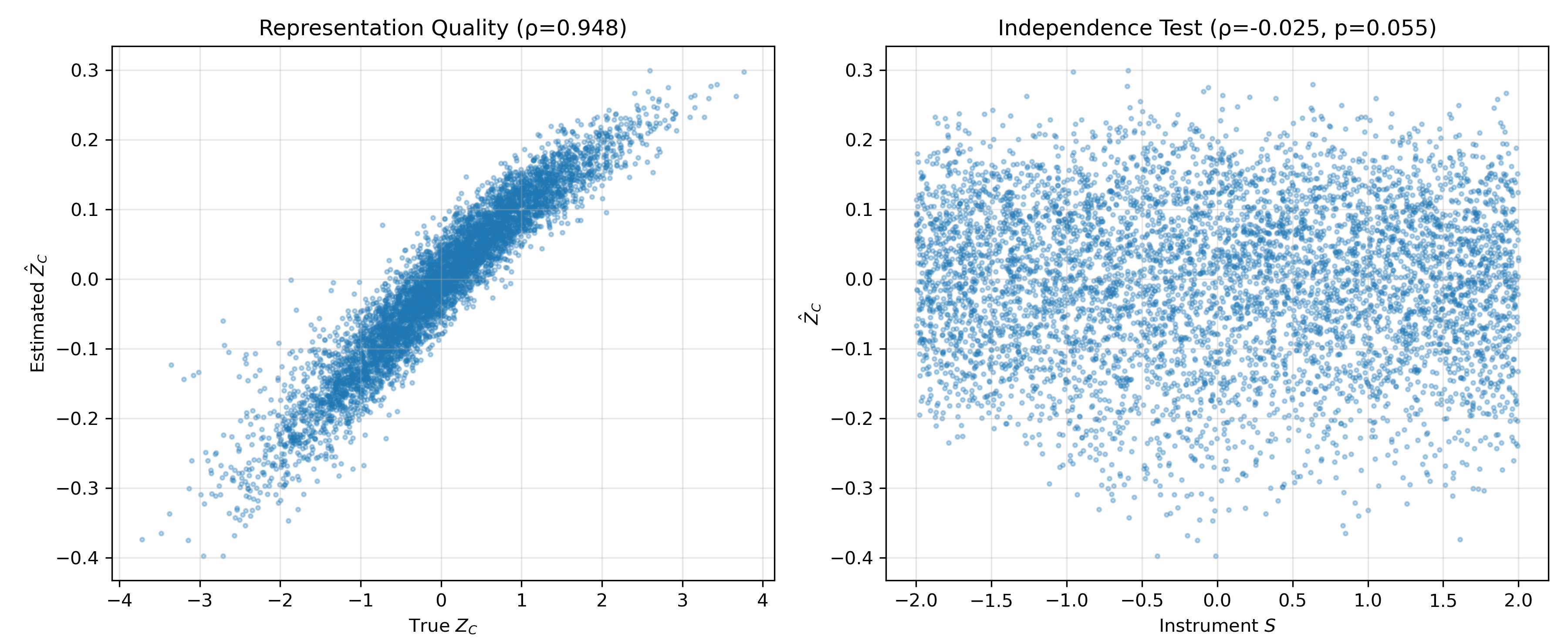}
    \caption{Nonlinear DGP.}
  \end{subfigure}

  \caption{Representation learning quality, demonstrated by correlation between true $Z_C$ and learned $\hat Z_C$, and the dependence between learned $\hat Z_C$ and instrument $S$.}
  \label{fig:rep}
\end{figure}

\subsubsection{Benchmarking ATE estimation}
Within the triple machine learning (TML, i.e., IV-based representation learning) framework, we can implement different estimators in the final estimation fold (including outcome regression, IPW, and double robust estimators). We apply weight clipping to both the IPW and doubly robust estimators to ensure numerical stability. We benchmarked the ATE estimation ($A=1$ vs $A=0$). We evaluated different variants of TML estimators, against 2SLS baselines over 100 replications, using bias and MSE (bias$^2+$SE$^2$) in each replication as the evaluation metrics, shown in Fig.\ref{fig:ATE}. We can find for continuous treatments, different estimators within the IV-based representation learning framework all significantly outperform 2SLS baseline. Even in settings with a binary treatment and linear effects, the performance of our nonparametric method is comparable to that of parametric approaches like 2SLS.
\begin{figure}[!htbp]
    \centering
  \begin{subfigure}{0.49\textwidth}
    \centering
    \includegraphics[width=\linewidth]{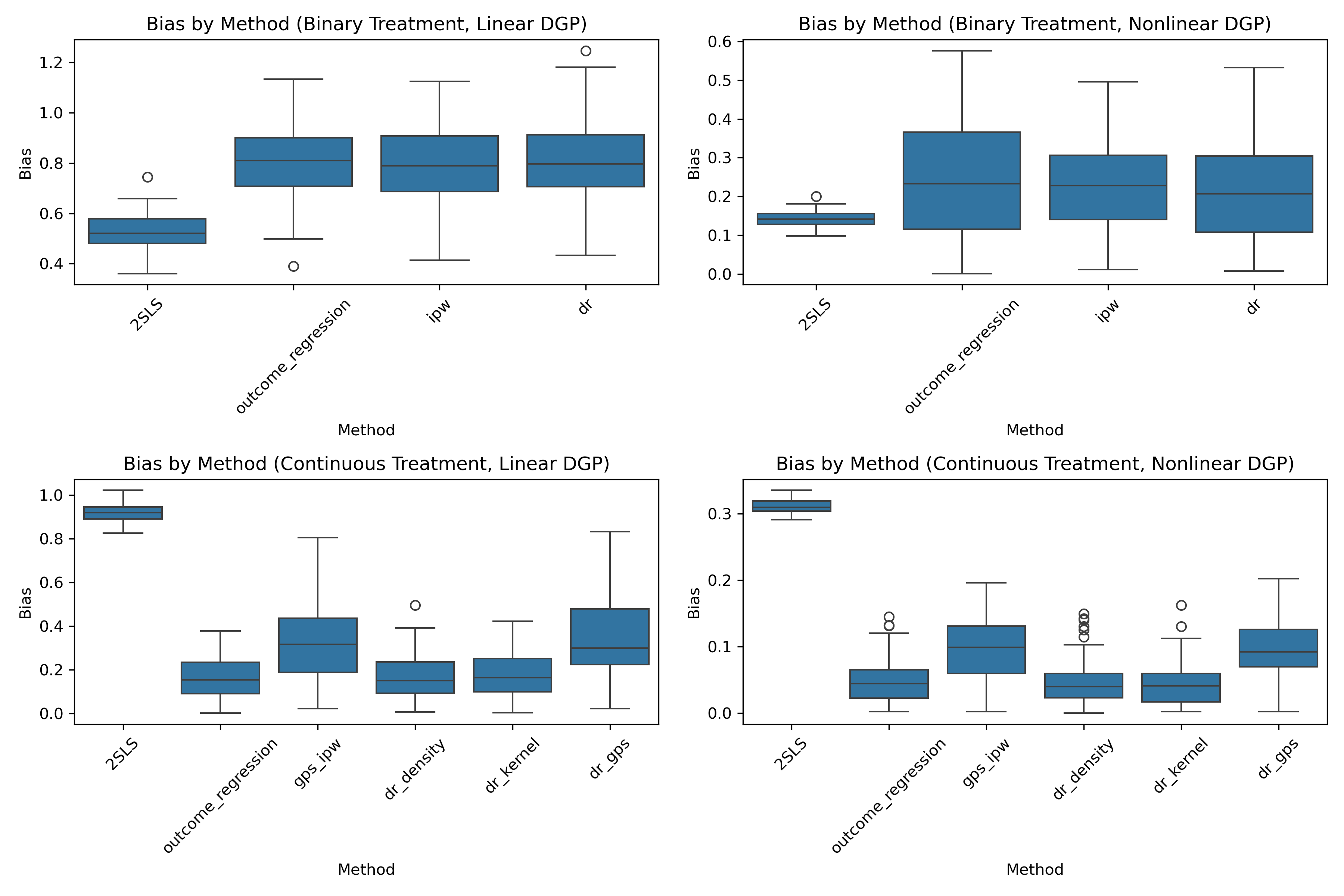}
    \caption{Boxplots of bias in each replicate of different estimators under different DGPs}
  \end{subfigure}
  \hfill
  \begin{subfigure}{0.49\textwidth}
    \centering
    \includegraphics[width=\linewidth]{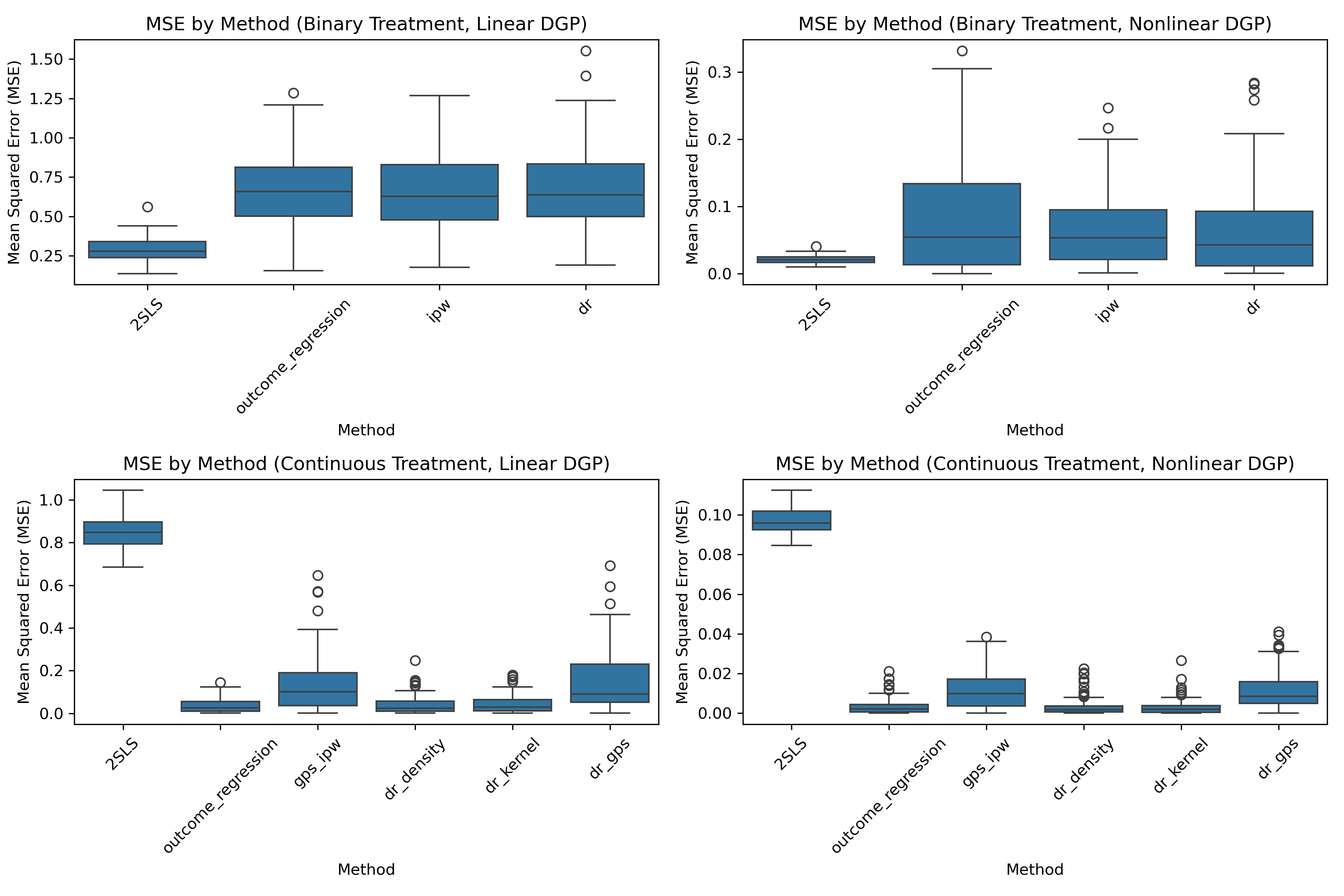}
    \caption{Boxplots of MSE in each replicate of different estimators under different DGPs}
  \end{subfigure}

  \caption{Simulation results of ATE estimation with unmeasured confounding. We evaluate 2SLS baseline with different estimators (outcome regression, IPW, and double robust) within triple machine learning (TML) framework.}
  \label{fig:ATE}
\end{figure}

\subsubsection{Dose-response curve estimation}
To evaluate the estimator's capability in recovering the full structural relationship between the treatment and the outcome, we conducted a dose-response curve estimation experiment with continuous treatment regime. Shown in Fig.\ref{fig:dose}a, TML can not only provide accurate estimates of the average treatment effect, but also recover the dose–response function $E[Y(a)]$ under continuous treatment assignment, with particularly strong performance observed for the kernel-based doubly robust estimator, especially in the treatment region with high observation density (Fig.\ref{fig:dose}b).

\begin{figure}[!htbp]
  \centering
  \begin{subfigure}{0.7\textwidth}
    \centering
    \includegraphics[width=\linewidth]{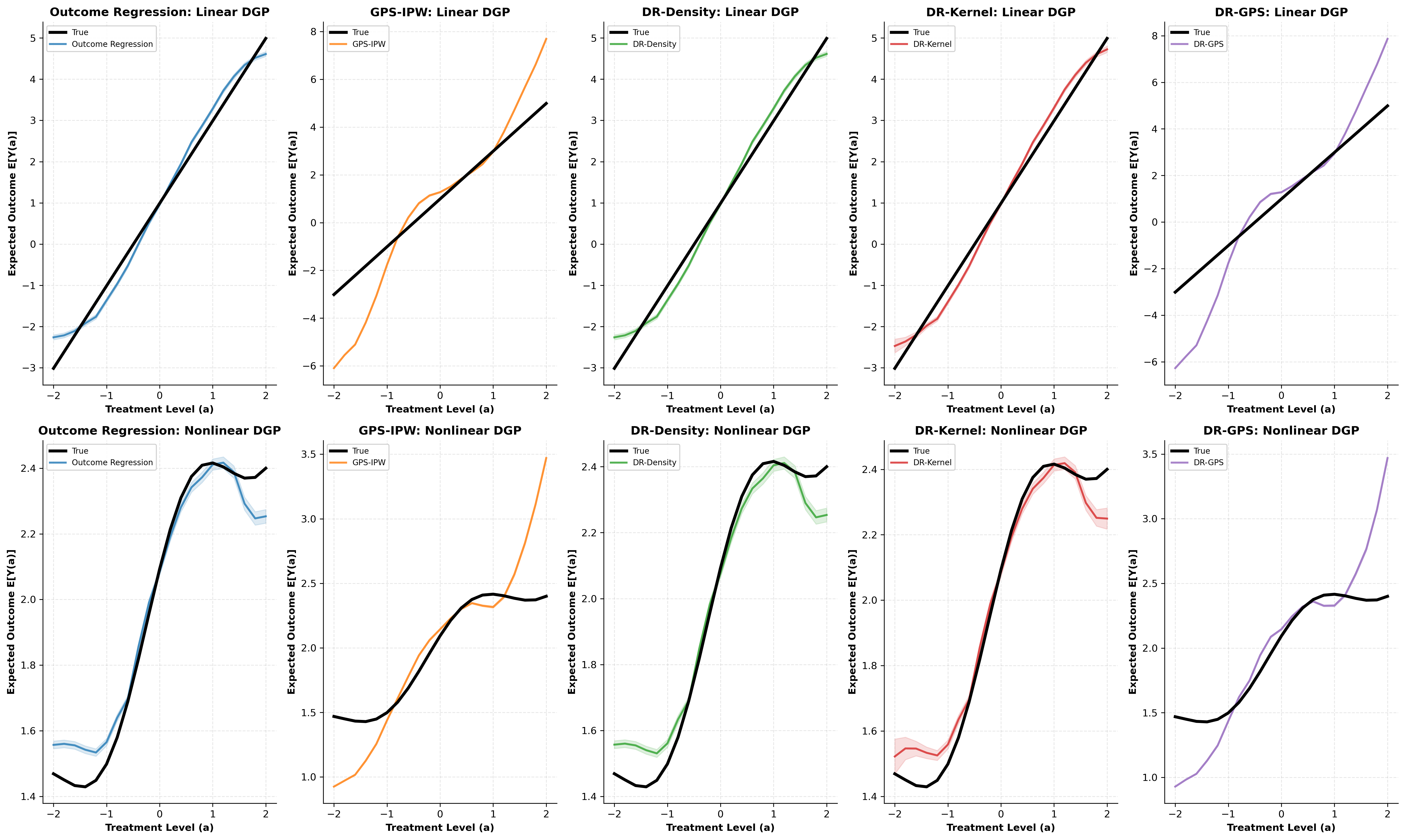}
    \caption{Dose response ($\E[Y(a)]$) curves estimated via different TML estimators.}
  \end{subfigure}
  \hfill
  \begin{subfigure}{0.25\textwidth}
    \centering
    \begin{subfigure}{\linewidth}
      \centering
      \includegraphics[width=\linewidth]{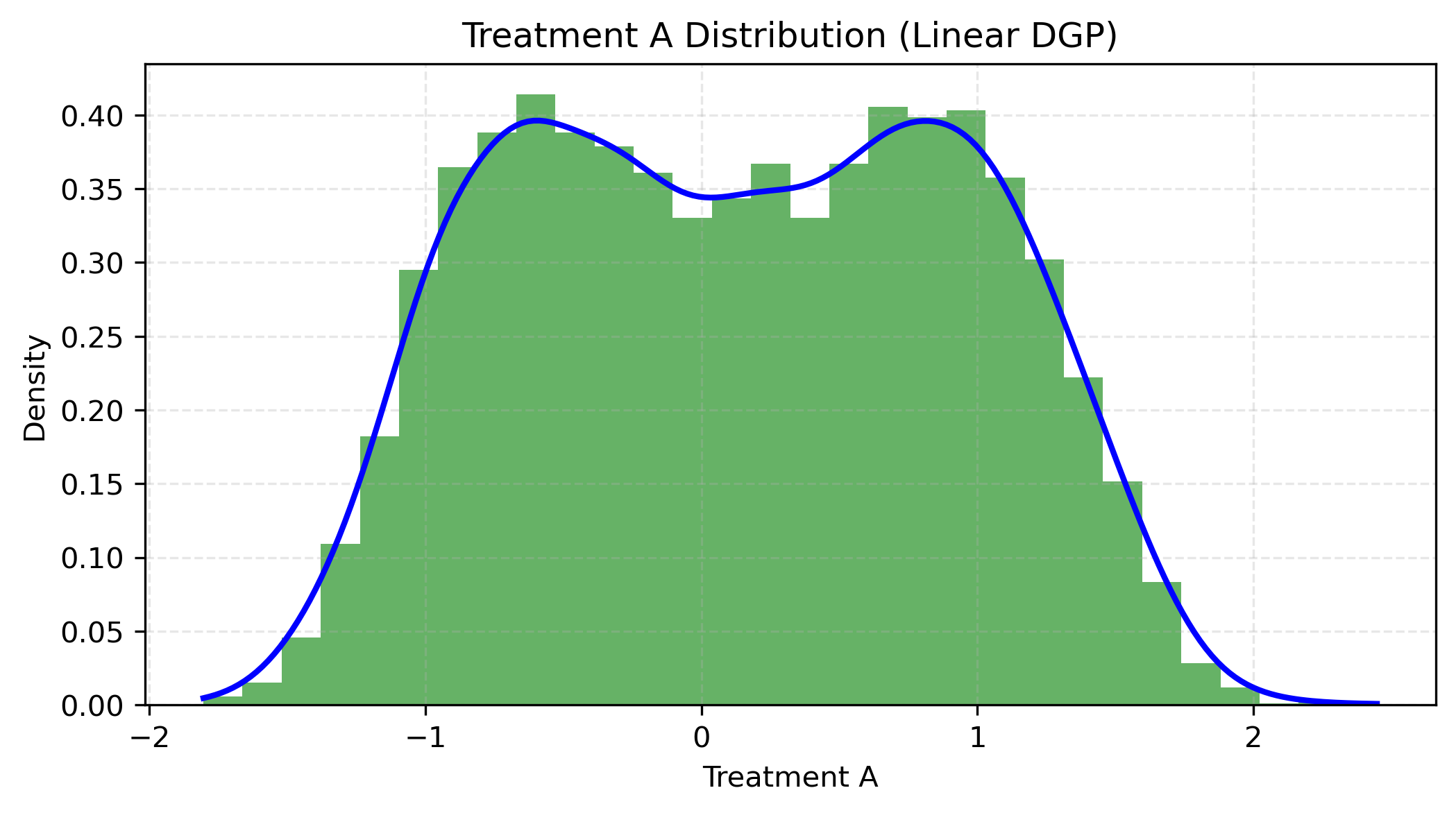}
    \end{subfigure}

    \vspace{0.6em}

    \begin{subfigure}{\linewidth}
      \centering
      \includegraphics[width=\linewidth]{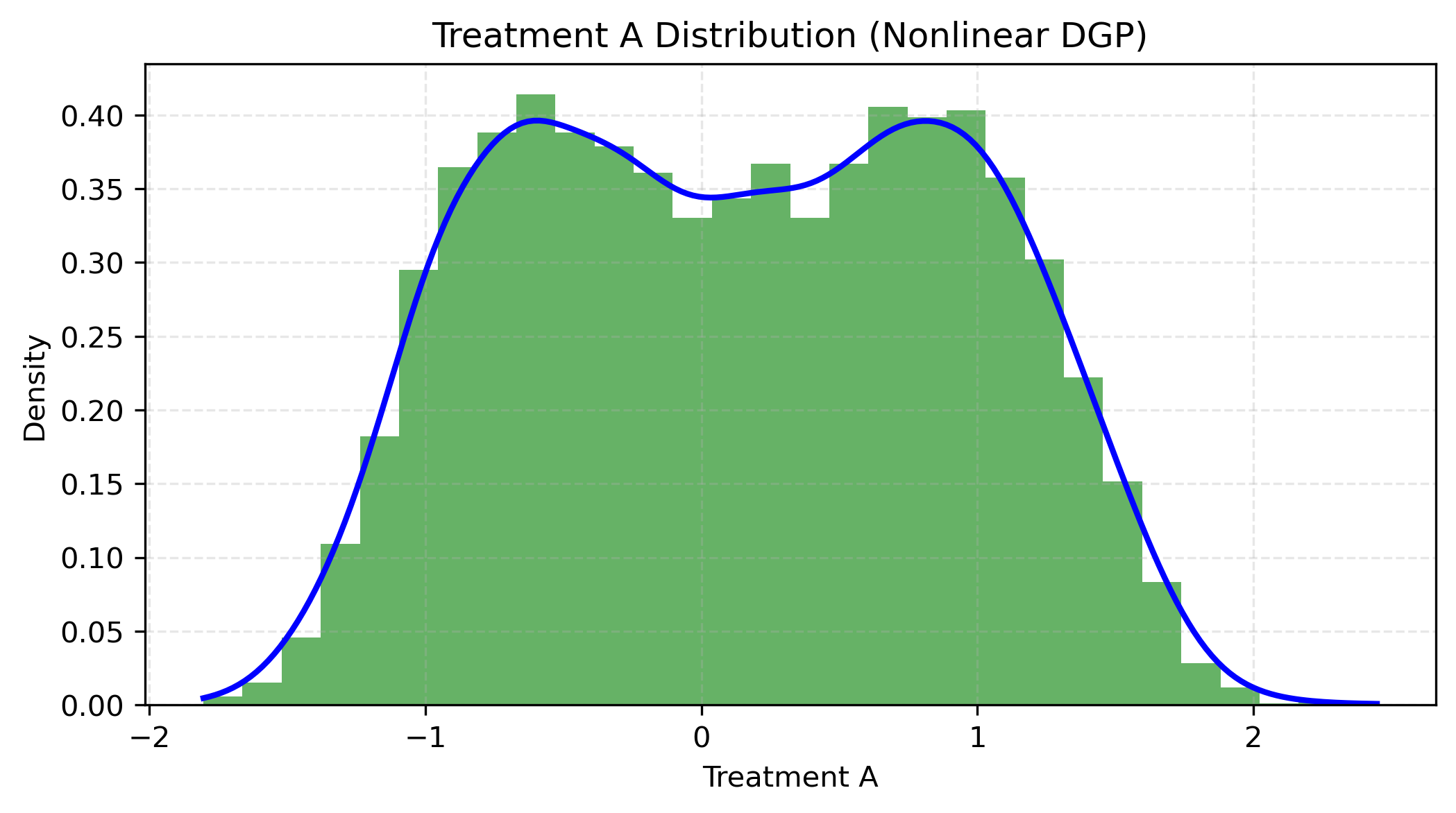}
    \end{subfigure}
  \caption{Observed treatment distribution in simulation}
  \end{subfigure}

  \caption{Dose-response curve estimation}
  \label{fig:dose}
\end{figure}

\section{Real-world analysis}\label{sec:app}
We illustrate the proposed methodology by analyzing the demand for cigarettes using the \texttt{CigarettesSW} dataset from the R package \texttt{AER} \cite{kleiber2020package}. This dataset comprises panel data for 48 U.S. states over the period 1985--1995. The primary objective is to estimate the price elasticity of demand while accounting for unobserved heterogeneity in state-level consumption behaviors.

Let $Y$ denote the logarithm of cigarette consumption (packs per capita) and $A$ be the logarithm of the real price. The relationship between price and consumption is confounded by unobserved state-specific factors, $Z_C$, such as regional health awareness and cultural habits. To address this endogeneity, we employ the logarithm of the average excise tax as the instrument $S$, which satisfies the strong relevance condition through its direct impact on price and is plausibly exogenous to individual preferences.

\begin{figure}[!hbtp]
    \centering
    \includegraphics[width=\linewidth]{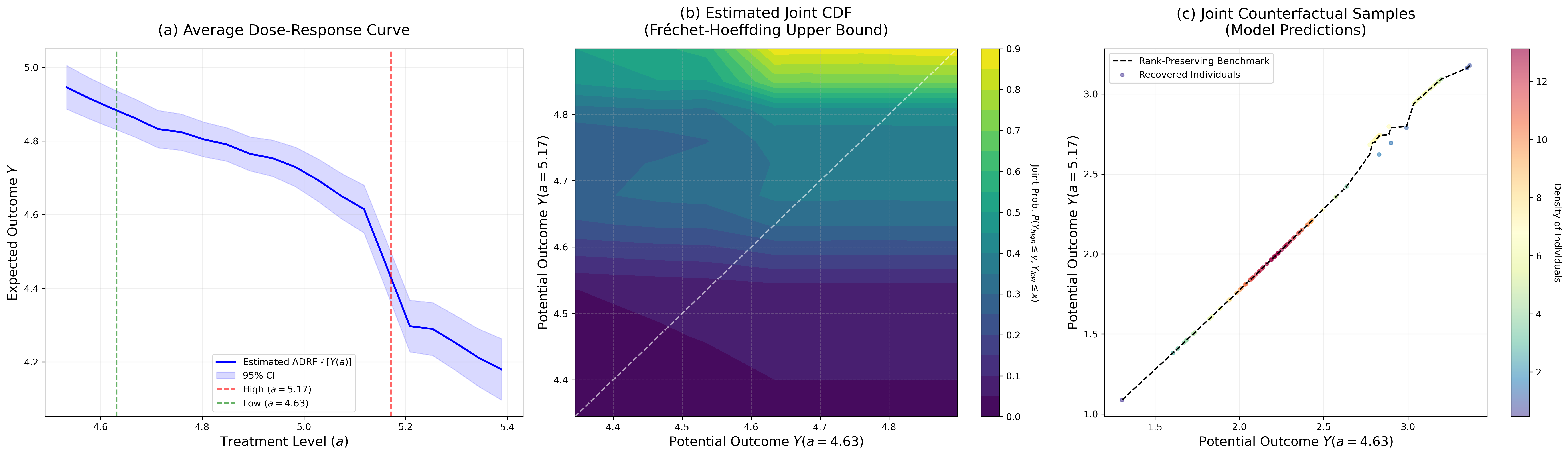}
    \caption{Causal analysis of cigarette demand. 
(a) Estimated Average Dose-Response Function with 95\% pointwise confidence intervals. The convex shape indicates that price elasticity increases (becomes more negative) at higher price levels. 
(b) The joint CDF surface estimated via the proposed smooth estimator, targeting the FH upper bound. The concentration of probability mass along the diagonal visualizes the theoretical limit of rank preservation. 
(c) Scatter plot of individual counterfactual pairs $(\hat{Y}(a_{\mathrm{low}}), \hat{Y}(a_{\mathrm{high}}))$ reconstructed by the fitted outcome model for $a_{\mathrm{low}}=4.63$ and $a_{\mathrm{high}}=5.17$. By fixing the learned latent heterogeneity $\hat{z}_{c,i}$ for each state, the strict alignment with the identity benchmark (dashed line) demonstrates that the treatment induces a homogeneous structural shift across the population.}
\label{fig:app-result}
\end{figure}

We apply the proposed TML framework to recover the latent confounding structure $Z_C$ and estimate the causal mechanism. Figure~\ref{fig:app-result} summarizes the estimated average and distributional effects based on the learned representation. Panel (a) presents the estimated Average Dose-Response Function, $\mathbb{E}[Y(a)]$, across the observed support of log prices. The curve exhibits a monotonic decrease, consistent with the fundamental law of demand. The narrow 95\% pointwise confidence intervals indicate precise estimation of the mean effect. Notably, the response exhibits non-linearity: the slope steepens for prices above $a > 5.1$, suggesting an increasing price elasticity in the upper price range. This implies that consumption becomes more sensitive to marginal price increases when prices are already elevated.

Beyond the average effect, Panels (b) and (c) examine the joint dependence structure of potential outcomes. Panel (b) visualizes the joint CDF estimated via the proposed smooth estimator, specifically targeting the FH upper bound $P(Y(a_{\mathrm{high}}) \leq y_1, Y(a_{\mathrm{low}}) \leq y_0)$. This estimation reflects the limit of perfect rank preservation. Panel (c) displays the empirical counterpart at the individual level. Since the true counterfactuals are unobserved, we reconstruct them using the fitted outcome model $\hat{Y}(a, z_c)$. For each state $i$, we fix its learned latent representation $\hat{z}_{c,i}$ and compute the predicted potential outcomes under two quartile values of observed price, $a_{\mathrm{low}}=4.63$ and $a_{\mathrm{high}}=5.17$ using the outcome model we trained. The resulting scatter plot of these predicted pairs $(\hat{Y}(a_{\mathrm{low}}, \hat{z}_{c,i}), \hat{Y}(a_{\mathrm{high}}, \hat{z}_{c,i}))$ aligns strictly with the rank-preserving benchmark (dashed line), confirming that the latent confounder $Z_C$ induces a comonotonic dependence structure.

Crucially, the linearity observed in Panel (c) on the logarithmic scale implies a homogeneous relative treatment effect. While the latent confounder $Z_C$ captures significant heterogeneity in baseline consumption levels (intercepts), the treatment effect manifests as a uniform structural shift across the distribution. This suggests that price interventions dampen consumption proportionally for both heavy and light smokers, preserving the rank ordering of states by consumption intensity.

\section{Discussion}\label{sec:discuss}
This paper has developed a unified framework for the identification and estimation of counterfactual distributions, addressing two fundamental challenges: the bounding of joint distributions under observed confounding and the recovery of marginal distributions under unobserved confounding. By bridging copula theory with causal representation learning, we provide a pathway from identifying population-level effects to bounding individual-level counterfactual dependencies.

Our analysis of conditional copulas underscores the informative value of covariates in tightening the FH bounds. The upper bound, corresponding to conditional rank preservation, serves as a particularly pragmatic benchmark for individualized treatment effect estimation. We proposed two estimation strategies, the direct estimator under margin conditions and the smooth log-sum-exp approximation, to handle the non-differentiable influence function and highlight a fundamental bias-variance trade-off.

In the context of unmeasured confounding, our integration of instrumental variables with causal representation learning offers a nonparametric alternative to classical SEMs. Unlike traditional IV methods that focus on local average treatment effects or rely on linearity, our approach leverages the instrument to identify the latent confounding subspace itself. The triple machine learning procedure extends the double machine learning paradigm to account for the representation learning stage. A key theoretical finding of this work is the characterization of the variance inflation induced by the estimated representation. As shown in Theorem \ref{thm:dr-rliv-summary}, unless the representation learner achieves super-convergence rates, a formidable requirement for deep neural networks standard errors must be corrected to account for the first-stage estimation uncertainty.

Several limitations of the current framework warrant further investigation. First, the identification of the latent confounding space relies on injectivity and completeness conditions, and relaxing these conditions to allow for partial identification of the latent space would be a valuable extension. Meanwhile, while we propose a HSIC regularized VAE-based algorithm with identifiability guarantees, finite sample size, constrained function approximators, and the non-convex nature of such optimization problems still poses challenges in real-world estimation \cite{hyvarinen2023nonlinear}. 

Furthermore, a practical challenge arises in the inference stage regarding the representation learning correction. The variance of the proposed estimator is inflated by the estimation error of the representation $\widehat{Z}_C$. However, quantifying this inflation explicitly requires the gradients of the invertible mapping $\psi$ linking the learned and true representations, which is unfortunately analytically unknown and implicitly defined by the neural network training dynamics, making estimating the correction term $\mathrm{IF}_{\phi, corr}$ non-trivial. Relying on uncorrected variance estimators may lead an inflation of Type I error rates in hypothesis testing. Developing feasible methods, such as Hessian approximation \cite{koh2017understanding} or ensemble-based quantification \cite{wager2018estimation}, to validly capture this uncertainty remains an open and critical area for future research \cite{gawlikowski2023survey}.

In summary, this work bridges the gap between classical semiparametric theory and modern representation learning. By formally characterizing the identification limits of latent confounding through instrumental variables and establishing the asymptotic distribution of the ``triple machine learning'' estimator, we provide a principled alternative to \textit{ad hoc} deep causal estimation. Simultaneously, our copula-based bounds refine the understanding of treatment effect heterogeneity beyond marginal aggregates, offering a rigorous tool for assessing individual-level risks. Collectively, these contributions lay the groundwork for a more robust statistical foundation for causal inference with high-dimensional, unstructured data.

\section*{Acknowledgment}
We gratefully acknowledge Prof. Edward Kennedy, JungHo Lee, and Ignavier Ng for their valuable comments and suggestions on the manuscript.

\bibliographystyle{agsm}

\bibliography{refs}

\newpage
\appendix
\section*{Technical Appendix}\label{sec:append}
\input{appendix.tex}

\end{document}

%% file: appendix.tex
\section{Proof of Theorem \ref{thm:bound}}\label{proof:bound}
\begin{proof}
\noindent\textbf{1) Proof for the Upper Bound ($U \le U_\text{marg}$)}
\nopagebreak

Define the function $g: \mathbb{R}^2 \to \mathbb{R}$ as $g(a, b) = \min(a, b)$. The function $g$ is the pointwise minimum of two linear (and thus concave) functions, $f_1(a,b)=a$ and $f_2(a,b)=b$. Therefore, $g$ is a \textbf{concave function}.

By Jensen's Inequality for concave functions, for any random vector $Y$, we have $\E[g(Y)] \le g(\E[Y])$.
Let the random vector be $(\theta_1(X), \theta_0(X))$. Applying the inequality:
\[
\E_X\left[ g(\theta_1(X), \theta_0(X)) \right] \le g\left( \E_X[\theta_1(X)], \E_X[\theta_0(X)] \right)
\]
Substituting the definitions of $g$, $\theta_a$, $U$, and $U_\text{marg}$:
\[
\underbrace{\E_X\bigl[\min\{\theta_1(X), \theta_0(X)\}\bigr]}_{U(y_1, y_0)} \le \min\bigl\{\E_X[\theta_1(X)], \, \E_X[\theta_0(X)]\bigr\}
\]
Since $F_{Y(a)}(y_a) = \E_X[\theta_a(X)]$, the right-hand side is
\[
\min\bigl\{F_{Y(1)}(y_1), \, F_{Y(0)}(y_0)\bigr\} = U_\text{marg}(y_1, y_0)
\]
Thus, we have shown that $U(y_1, y_0) \le U_\text{marg}(y_1, y_0)$.

\medskip
\noindent\textbf{2) Proof for the Lower Bound ($L_\text{marg} \le L$)}
\nopagebreak

Define the function $h: \mathbb{R}^2 \to \mathbb{R}$ as $h(a, b) = \max(a + b - 1, 0)$. The function $h$ is the pointwise maximum of two linear (and thus convex) functions, $f_1(a,b)=a+b-1$ and $f_2(a,b)=0$. Therefore, $h$ is a \textbf{convex function}.

By Jensen's Inequality for convex functions, for any random vector $Y$, we have $\E[h(Y)] \ge h(\E[Y])$.
Let the random vector again be $(\theta_1(X), \theta_0(X))$. Applying the inequality:
\[
\E_X\left[ h(\theta_1(X), \theta_0(X)) \right] \ge h\left( \E_X[\theta_1(X)], \E_X[\theta_0(X)] \right)
\]
Substituting the definitions of $h$ and $\theta_a$
\[
\underbrace{\E_X\bigl[\max\{\theta_1(X) + \theta_0(X) - 1, \, 0\}\bigr]}_{L(y_1, y_0)} \ge \max\left\{ \E_X[\theta_1(X)] + \E_X[\theta_0(X)] - 1, \, 0 \right\}
\]
(We used the linearity of expectation inside the $\max$ function on the right-hand side).
Since $F_{Y(a)}(y_a) = \E_X[\theta_a(X)]$, the right-hand side is
\[
\max\left\{ F_{Y(1)}(y_1) + F_{Y(0)}(y_0) - 1, \, 0 \right\} = L_\text{marg}(y_1, y_0)
\]
Thus, we have shown that $L(y_1, y_0) \ge L_\text{marg}(y_1, y_0)$.

\medskip
We have established $L_\text{marg} \le L$ and $U \le U_\text{marg}$, which proves that the interval $\covbounds$ is a (potentially tighter) sub-interval of $\margbounds$. Equality holds if and only if the conditional CDFs $\theta_a(X)$ are constant almost surely with respect to $X$, implying that the covariates $X$ provide no additional information beyond the marginals.
\end{proof}

\section{Proof on efficient estimator of conditional FH bound}
We first derive the IF for the conditional CDF functional $\theta_a(\cdot)$, then use chain rule to obtain EIF for smooth functional $\phi$, and finally handle the nonsmooth $\min$ function using (i) smoothing and (ii) direct subgradient under margin condition.
\subsection{Influence function components under regular condition}
\subsubsection{EIF for $\theta_a(\cdot)$ and for $\E[\theta_a(X)]$}

Fix $y\in\mathbb{R}$ and $a\in\{0,1\}$. Define the functional on the full
distribution $P$
\[
\Theta_a(P)(x) := \theta_a(x) = \P(Y\le y\mid A=a,X=x).
\]
We are interested in the scalar functional $\Psi_a(P):=\E_X[\theta_a(X)]$.
To obtain the pathwise derivative (influence function) we follow the standard
parametric-submodel / Gateaux-derivative route.

Let $\{P_\varepsilon:\varepsilon\in(-\varepsilon_0,\varepsilon_0)\}$ be any
regular parametric submodel with score $s=o(1)$ at $\varepsilon=0$, i.e.
$dP_\varepsilon/d\P = 1+\varepsilon s + o(\varepsilon)$ with $\E_P[s]=0$.
Write $\theta_{a,\varepsilon}(x)$ for the conditional CDF under $P_\varepsilon$.
We compute
\[
\frac{d}{d\varepsilon}\Big|_{\varepsilon=0}\Psi_a(P_\varepsilon)
= \frac{d}{d\varepsilon}\Big|_{\varepsilon=0}\E_{P_\varepsilon}\big[\theta_{a,\varepsilon}(X)\big].
\]
Differentiate using product rule
\[
\frac{d}{d\varepsilon}\E_{P_\varepsilon}[\theta_{a,\varepsilon}(X)]
=
\E_P\big[\dot\theta_a(X)\big] + \E_P\big[\theta_a(X) s(O)\big],
\]
where $\dot\theta_a(x):=\frac{d}{d\varepsilon}\big|_{\varepsilon=0}\theta_{a,\varepsilon}(x)$.
We now compute $\dot\theta_a(x)$ by differentiating the conditional CDF
\[
\theta_{a,\varepsilon}(x) = P_\varepsilon(Y\le y\mid A=a,X=x)
= \frac{\E_{P_\varepsilon}\big[\1\{A=a\}\1\{Y\le y\}\mid X=x\big]}{P_\varepsilon(A=a\mid X=x)}.
\]
Differentiate at $\varepsilon=0$; denote $\pi_a(x)=\P(A=a\mid X=x)$. Using
quotient rule we obtain
\begin{equation}\label{eq:theta_dot}
\dot\theta_a(x)
= \frac{1}{\pi_a(x)}\E_P\big[(\1\{A=a\}\1\{Y\le y\}) s(O) \mid X=x\big]
- \frac{\theta_a(x)}{\pi_a(x)}\E_P\big[\1\{A=a\} s(O)\mid X=x\big].
\end{equation}
Multiply both sides by the marginal density of $X$ and integrate to get the
Gateaux derivative of $\Psi_a$. Rearranging and using standard score
calculations yields that the influence function for $\Psi_a$ is
\begin{equation}\label{eq:IF_theta_a_appendix}
\IF_a(\Psi_a(P)) \;=\; \theta_a(X)-\E[\theta_a(X)] \;+\;
\frac{\1\{A=a\}}{\pi_a(X)}\big(\1\{Y\le y\}-\theta_a(X)\big).
\end{equation}
We can check $\E[\IF_a(\Psi_a(P))]=0$ and that the pathwise derivative equals
$\E[s(O)\IF_a(\Psi_a(P))]$ for all scores $s$, which confirms \eqref{eq:IF_theta_a_appendix}
is the canonical gradient for $\Psi_a$.

\subsubsection{Chain rule: EIF for smooth $\phi$}

Let $\phi:\mathbb{R}^2\to\mathbb{R}$ be continuously differentiable. Define
\[
\Psi(P)=\E_X\big[\phi\big(\theta_0(X),\theta_1(X)\big)\big].
\]
Consider a parametric submodel $P_\varepsilon$ with score $s$ and let
$\theta_{a,\varepsilon}(x)$ be the conditional CDF under $P_\varepsilon$.
Differentiate
\[
\frac{d}{d\varepsilon}\Big|_{\varepsilon=0}\Psi(P_\varepsilon)
=
\E_P\Big[ \sum_{a=0}^1 \partial_a\phi(\theta_0(X),\theta_1(X)) \dot\theta_a(X) \Big]
+ \E_P\big[ \phi(\theta_0,\theta_1) s(O)\big].
\]
Using \eqref{eq:IF_theta_a_appendix} which gives the derivative of
$\E[\theta_a(X)]$, we can rewrite the above as
\[
\frac{d}{d\varepsilon}\Big|_{\varepsilon=0}\Psi(P_\varepsilon)
= \E_P\big[ s(O)\cdot \IF(\Psi(P))\big],
\]
with
\begin{equation}\label{eq:IF_smooth_appendix}
\IF(\Psi(P)) \;=\; \sum_{a=0}^1 \partial_a\phi(\theta_0(X),\theta_1(X))\frac{\1\{A=a\}}{\pi_a(X)}(\1\{Y\le y\}-\theta_a(X))+\phi(\theta_0(X),\theta_1(X))-\Psi(P).
\end{equation}
Again we check $\E_P[\IF(\Psi(P))]=0$, and for any score $s$,
$\frac{d}{d\varepsilon}\Psi(P_\varepsilon)|_{\varepsilon=0}=\int(\IF\cdot s)\,dP$, verifying \eqref{eq:IF_smooth_appendix}
is the canonical gradient.

We now treat $\phi(u,v)=\min\{u,v\}$. Since $\min$ is not differentiable
on the diagonal $u=v$, we consider two strategies: Direct Estimation under a Margin Condition or approximation via a smooth function.

\subsection{Proof of asymptotic property of direct estimator under margin condition}\label{proof:direct}
We first prove the statistical properties of direct estimator (Theorem \ref{thm:direct}) by introducing a polynomial margin condition (Assumption \ref{ass:polymargin}).
\begin{proof}
We follow the standard semi-parametric argument.

\noindent\textbf{Step 1. Oracle representation and pathwise derivative (canonical gradient).}

Assume for the moment the selector \(d(x)\) is known (``oracle'' case). Then
\[
\Psi(P)=\E\big[ \1\{d(X)=0\}\theta_0(X) + \1\{d(X)=1\}\theta_1(X)\big].
\]
Let $\{P_\varepsilon:\varepsilon\}$ be a regular parametric submodel with score function $s(O)$ at $\varepsilon=0$ so that
\[
\frac{d}{d\varepsilon}\Big|_{\varepsilon=0} \log\frac{dP_\varepsilon}{dP} = s(O),\qquad \E[s(O)]=0.
\]
Differentiate $\Psi(P_\varepsilon)$ using product rule and the conditional structure.
We compute the derivative of each term. Fix \(a\in\{0,1\}\). For the term
$\E[\1\{d(X)=a\}\theta_{a,\varepsilon}(X)]$ we have
\[
\frac{d}{d\varepsilon}\Big|_{\varepsilon=0}\E_{P_\varepsilon}[\1\{d(X)=a\}\theta_{a,\varepsilon}(X)]
= \E\big[\1\{d(X)=a\}\dot\theta_a(X)\big] + \E\big[\1\{d(X)=a\}\theta_a(X)s(O)\big],
\]
where $\dot\theta_a(x)=\frac{d}{d\varepsilon}\big|_{\varepsilon=0}\theta_{a,\varepsilon}(x)$.
Using derived equation \ref{eq:theta_dot} for $\dot\theta_a(x)$
\[
\dot\theta_a(x) = \frac{1}{\pi_a(x)} \E\big[(\1\{A=a\}\1\{Y\le y\})\,s(O)\mid X=x\big]
- \frac{\theta_a(x)}{\pi_a(x)} \E\big[\1\{A=a\}\,s(O)\mid X=x\big],
\]
and integrating the pathwise derivative of $\Psi$ along the score \(s\) $\frac{d}{d\varepsilon}\Big|_{\varepsilon=0}\Psi(P_\varepsilon) = \E\big[ s(O)\, \IF_{\mathrm{oracle}}(\Psi(P;d))\big]$, we obtain the (centered) oracle influence function
\begin{equation}\label{eq:IF_oracle_detailed}
\IF_{\mathrm{oracle}}(\Psi(P;d))
= \sum_{a=0}^1 \1\{d(X)=a\}\Big( \frac{\1\{A=a\}}{\pi_a(X)}(\1\{Y\le y\}-\theta_a(X)) + \theta_a(X)-\Psi(P)\Big).
\end{equation}
We can check $\E[\IF_{\mathrm{oracle}}]=0$. Thus \eqref{eq:IF_oracle_detailed} is the canonical gradient under the oracle selector.

\medskip

\noindent\textbf{Step 2. Feasible estimator and von-Mises decomposition.}

In practice $d(x)$ is unknown; replace it by the plug-in selector in a separated independent data (by data splitting or cross-fitting)
\[
\hat d(x) = \arg\min_{a}\hat\theta_a(x).
\]
We denote the uncentered influence function corresponding to $\IF_{\mathrm{oracle}}+\Psi(P)$ as
\[
\varphi(O; P, d) =\sum_{a=0}^1 \1\{d(X)=a\}\Big( \frac{\1\{A=a\}}{\pi_a(X)}(\1\{Y\le y\}-\theta_a(X)) + \theta_a(X)\Big),
\]
plug-in evaluated at true selector
\[
\varphi(O;\hat P,d)=\sum_a \1\{d(X)=a\}\Big(
\hat\theta_a(X)
+ \tfrac{\1\{A=a\}}{\hat\pi_a(X)}(\1\{Y\le y\}-\hat\theta_a(X))\Big),
\]
and $\varphi(O;\hat P,\hat d)$ denotes the same expression with $\hat d$ in place of $d$, then define the feasible doubly-robust estimator
\begin{equation}\label{eq:hatPsi_def}
\hat\Psi = \Pn[\varphi(O;\hat P,\hat d)]
= \Pn\Big[\sum_{a=0}^1 \1\{\hat d(X)=a\} \Big(
\hat\theta_a(X) + \frac{\1\{A=a\}}{\hat\pi_a(X)}(\1\{Y\le y\}-\hat\theta_a(X))\Big)\Big].
\end{equation}
We compare $\hat\Psi$ to $\Psi$ by adding and subtracting the oracle influence function
\begin{align}
    \hat\Psi - \Psi(P)
    &= \Pn[\varphi(O; \hat P, \hat d)] - \P[\varphi(O; P, d)] \\
    &= (\Pn-\P)\varphi(O; P, d) + (\Pn-\P)[\varphi(O; \hat P, \hat d)-\varphi(O; P, d)] + \P[\varphi(O; \hat P, \hat d)-\varphi(O; P, d)] , \nonumber\\
    &=S+R_1+R_2,\nonumber
\end{align}
We bound $R_2$ by separating nuisance and selector errors
\begin{align*}
\P\left[\varphi(O;\hat P,\hat d)-\varphi(O;P,d)\right]
&= \underbrace{\P\left[\varphi(O;\hat P,d)-\varphi(O;P,d)\right]}_{\text{nuisance error } B_\text{nuis}}
+ \underbrace{\P\left[\varphi(O;\hat P,\hat d)-\varphi(O;\hat P,d)\right]}_{\text{selector error } B_\text{sel}}.
\end{align*}
\medskip

\noindent\textbf{Step 3. Bound for the nuisance remainder \(B_{\mathrm{nuis}}\).}
We first aim to bound $B_{\mathrm{nuis}} = \P\left[\varphi(O;\hat P,d)-\varphi(O;P,d)\right]$.
Recall that $\P[\cdot]$ denotes the expectation $\E_O[\cdot]$. We use the law of iterated expectations, $\E_O[\cdot] = \E_X\left[ \E_{A,Y|X}\left[ \cdot \mid X \right] \right]$, and first compute the inner conditional expectation.

For $a \in \{0,1\}$, let $\psi_a(X; \hat P)$ be the conditional expectation of the $a$-th component
\begin{align*}
    \psi_a(X; \hat P) &:= \E_{A,Y|X}\left[ \varphi(O;\hat P,a) \mid X \right] \\
    &= \E_{A,Y|X}\left[ \hat\theta_a(X) + \frac{\1\{A=a\}}{\hat\pi_a(X)}(\1\{Y\le y\}-\hat\theta_a(X)) \mid X \right] \\
    &= \hat\theta_a(X) + \frac{\E_{A|X}[\1\{A=a\} \mid X]}{\hat\pi_a(X)} \left( \E_{Y|A,X}[\1\{Y\le y\} \mid A=a, X] - \hat\theta_a(X) \right) \\
    &= \hat\theta_a(X) + \frac{\pi_a(X)}{\hat\pi_a(X)} \left( \theta_a(X) - \hat\theta_a(X) \right).
\end{align*}
The corresponding conditional expectation $\psi_a(X; P)$ evaluated at the true parameters $P$ is
\[
\psi_a(X; P) := \E_{A,Y|X}\left[ \varphi(O; P,a) \mid X \right]
= \theta_a(X) + \frac{\pi_a(X)}{\pi_a(X)} \left( \theta_a(X) - \theta_a(X) \right) = \theta_a(X).
\]
Now we can rewrite $B_{\mathrm{nuis}}$ by substituting these conditional expectations
\begin{align*}
    B_{\mathrm{nuis}} &= \E_X\left[ \E_{A,Y|X}\left[ \sum_{a=0}^1 \1\{d(X)=a\} \left( \varphi(O;\hat P,a) - \varphi(O;P,a) \right) \mid X \right] \right] \\
    &= \E_X\left[ \sum_{a=0}^1 \1\{d(X)=a\} \left( \psi_a(X; \hat P) - \psi_a(X; P) \right) \right].
\end{align*}
We compute the difference $\psi_a(X; \hat P) - \psi_a(X; P)$
\begin{align*}
    \psi_a(X; \hat P) - \psi_a(X; P) &= \left[ \hat\theta_a(X) + \frac{\pi_a(X)}{\hat\pi_a(X)} (\theta_a(X) - \hat\theta_a(X)) \right] - \theta_a(X) \\
    &= (\hat\theta_a(X) - \theta_a(X)) - \frac{\pi_a(X)}{\hat\pi_a(X)} (\hat\theta_a(X) - \theta_a(X)) \\
    &= (\hat\theta_a(X) - \theta_a(X)) \left( 1 - \frac{\pi_a(X)}{\hat\pi_a(X)} \right) \\
    &= (\hat\theta_a(X) - \theta_a(X)) \left( \frac{\hat\pi_a(X) - \pi_a(X)}{\hat\pi_a(X)} \right).
\end{align*}
This identity shows that the nuisance remainder is a product of the estimation errors in $\theta_a$ and $\pi_a$. This is the key ``Neyman-Orthogonal" or ``Doubly-Robust" structure, which ensures the first-order (linear) error terms cancel exactly.

Substituting this back into the expression for $B_{\mathrm{nuis}}$ yields
\[
B_{\mathrm{nuis}} = \E_X\left[ \sum_{a=0}^1 \1\{d(X)=a\} (\hat\theta_a(X) - \theta_a(X)) \left( \frac{\hat\pi_a(X) - \pi_a(X)}{\hat\pi_a(X)} \right) \right].
\]
We now bound this remainder. Assuming the estimators are bounded $\inf_x \hat\pi_a(x) \ge \underline\pi > 0$, by Cauchy-Schwarz, we have
\begin{align*}
    |B_{\mathrm{nuis}}| &\le \E_X\left[ \sum_{a=0}^1 \1\{d(X)=a\} \left| \hat\theta_a - \theta_a \right| \cdot \left| \frac{\hat\pi_a - \pi_a}{\hat\pi_a} \right| \right] \\
    &\lesssim \sum_{a=0}^1 \E_X\left[ \left| \hat\theta_a(X) - \theta_a(X) \right| \cdot \left| \hat\pi_a(X) - \pi_a(X) \right| \right] \\
    &\le \sum_{a=0}^1 \|\hat\theta_a - \theta_a\|_{L_2(P)} \cdot \|\hat\pi_a - \pi_a\|_{L_2(P)}.
\end{align*}
Thus, we have the tight, second-order bound
\begin{equation}\label{eq:B_nuis_bound}
    |B_{\mathrm{nuis}}| = O_p\Big( \sum_{a=0}^1 \|\hat\theta_a-\theta_a\|_{L_2(P)}\|\hat\pi_a-\pi_a\|_{L_2(P)} \Big).
\end{equation}
This bound shows that $B_{\mathrm{nuis}}=o_p(n^{-1/2})$ under the product-rate condition $\|\hat\theta_a-\theta_a\|_{L_2(P)}\|\hat\pi_a-\pi_a\|_{L_2(P)}=o_p(n^{-1/2})$.

\medskip

\noindent\textbf{Step 4. Bound for the selector remainder \(B_{\mathrm{sel}}\).}

The selector remainder arises because we use \(\hat d\) instead of \(d\). Note $\varphi(O;\hat P,a)=\hat\theta_a(X)
+ \tfrac{\1\{A=a\}}{\hat\pi_a(X)}(\1\{Y\le y\}-\hat\theta_a(X))$. For brevity we write $\Delta_d(X):=\1\{\hat d(X)\ne d(X)\}\in\{0,1\}$, $\Delta_\varphi(O;\hat P)=\varphi(O;\hat P,1)-\varphi(O;\hat P,0)$, $\Delta(X)=\theta_1(X)-\theta_0(X)$ and similarly $\hat\Delta(X)=\hat\theta_1(X)-\hat\theta_0(X)$. 
Using the pointwise identity
\begin{align}\label{eq:sel_decom}
    \varphi(O;\hat P,\hat d)-\varphi(O;\hat P,d)
    &=\sum_{a\in\{0,1\}}[\1(\hat d(X)=a)\varphi(O;\hat P,1)]-\sum_{a\in\{0,1\}}[\1(d(X)=a)\varphi(O;\hat P,1)]\,\nonumber\\
&= \big(\1\{\widehat d(X)=1\}-\1\{d(X)=1\}\big)\cdot\big(\varphi(O;\hat P,1)-\varphi(O;\hat P,0)\big)\,\nonumber\\
&=\Delta_d(X)\operatorname{sgn}(\hat d(X)-d(X)) \cdot\Delta_\varphi(O;\hat P)
\end{align}
Recall $\psi_a(X; \hat\theta,\hat\pi) := \E_{A,Y|X}\left[ \varphi(O;\hat P,a) \mid X \right] = \hat\theta_a(X) + \frac{\pi_a(X)}{\hat\pi_a(X)} \left( \theta_a(X) - \hat\theta_a(X) \right)$. Apply the Law of Iterated Expectations and substitute this into the expression for $B_\text{sel}$, we have
\begin{align*}
B_\text{sel} &= \P\left[\varphi(O;\hat P,\hat d)-\varphi(O;\hat P,d)\right] = \E_X\left[ \E_{A,Y|X}\left[ \varphi(O;\hat P,\hat d)-\varphi(O;\hat P,d) \mid X \right] \right]\\
&= \E_X\left[\Delta_d(X)\operatorname{sgn}(\hat d(X)-d(X)) \cdot (\psi_1(X; \hat\theta,\hat\pi) - \psi_0(X; \hat\theta,\hat\pi)) \right]\\
&\le \E_X\left[ \1\{\hat d \ne d\} \cdot |M^*(X)| \right],
\end{align*}
where $M^*(X) = \psi_1(X; \hat\theta,\hat\pi) - \psi_0(X; \hat\theta,\hat\pi)=\E_{A,Y|X}[\Delta_\varphi(O;\hat P)]$.
Assume the propensity estimators are uniformly bounded away from zero, $\inf_x \min_a \hat\pi_a(x) \ge \underline\pi>0$,
and note $|\pi_a|\le1$, $|\theta_a|\le 1$ and $|\hat\theta_a|\le 1$. Observe that pointwise
\[
|\psi_a| \le \max_{a}\Big|\hat\theta_a(X) + \frac{\pi_a(X)}{\hat\pi_a(X)} \left( \theta_a(X) - \hat\theta_a(X) \right)\Big|
\le 1 + \frac{1}{\underline\pi},
\]
so \(|M^*|\) is uniformly integrable and $\sup_X|M^*|<\infty$. Thus,
\[
| B_\text{sel} |
\le \P(\hat d(X)\ne d(X))\cdot \sup|M^*(X)|.
\]
Note that $\hat d(X)\ne d(X)$ implies that the sign of $\Delta(X)$ is flipped by estimation error large enough to overcome the gap. Indeed,
\[
\{\hat d(X)\ne d(X)\} \subseteq \left\{ |\Delta(X)| \le |\hat\Delta(X)-\Delta(X)| \right\}
= \left\{ |\Delta(X)| \le \|\hat\Delta(X)-\Delta(X)\|_\infty \right\}.
\]
Therefore, by Assumption~\ref{ass:polymargin} (polynomial margin of exponent $\alpha$),
\[
\P\big(\hat d(X)\ne d(X)\big) \le \P\Big(|\Delta(X)| \le \|\hat\Delta(X)-\Delta(X)\|_\infty\Big)
\lesssim \|\hat\Delta-\Delta\|_\infty^\alpha.
\]
We know
\[
\|\hat\Delta-\Delta\|_\infty \le \|\hat\theta_0-\theta_0\|_\infty + \|\hat\theta_1-\theta_1\|_\infty
\le 2\max_a\|\hat\theta_a-\theta_a\|_\infty.
\]
Hence
\[
\P(\hat d\ne d) \lesssim \Big(\max_a\|\hat\theta_a-\theta_a\|_\infty\Big)^\alpha.
\]
We then decompose 
\begin{align*}
|M^*(X)| &= \left|\left( \hat\theta_1 + \frac{\pi_1}{\hat\pi_1}(\theta_1 - \hat\theta_1) \right) - \left( \hat\theta_0 + \frac{\pi_0}{\hat\pi_0}(\theta_0 - \hat\theta_0) \right)\right| \\
&= \left|(\hat\theta_1 - \hat\theta_0) + \left[ \frac{\pi_1}{\hat\pi_1}(\theta_1 - \hat\theta_1) - \frac{\pi_0}{\hat\pi_0}(\theta_0 - \hat\theta_0) \right]\right| \\
&= |\hat\Delta(X) + R_{\text{IPW}}(X)|\\
&= |\Delta(X) + (\hat\Delta(X) - \Delta(X)) + R_{\text{IPW}}(X)| \\
& \leq |\Delta(X)| + |(\hat\Delta(X) - \Delta(X))| + |R_{\text{IPW}}(X)|.
\end{align*}
On the event $\hat{d}\ne d$, we know $|\Delta(X)| \le |\hat\Delta(X) - \Delta(X)|\lesssim\max_a\|\hat\theta_a-\theta_a\|_\infty$. Also, $|R_{\text{IPW}}(X)| \le \left|\frac{\pi_1}{\hat\pi_1}(\theta_1 - \hat\theta_1)\right| + \left|\frac{\pi_0}{\hat\pi_0}(\theta_0 - \hat\theta_0)\right| \lesssim \max_a\|\hat\theta_a-\theta_a\|_\infty$. Hence we have $\sup|M^*(X)| \lesssim \max_a\|\hat\theta_a-\theta_a\|_\infty$.

Inserting the margin bound for \(\P(\hat d\ne d)\) yields
\begin{equation}\label{eq:B_sel_bound}
|B_{\mathrm{sel}}|
\lesssim \Big(\max_a\|\hat\theta_a-\theta_a\|_\infty\Big)^{1+\alpha}.
\end{equation}
This is the key selector-bias bound: the cost of using the plug-in selector is controlled by a $(1+\alpha)$ power of the sup-norm estimation error.

\medskip

\noindent\textbf{Step 5. Bound the empirical process term}
We now bound the empirical process term
\[
R_1 = (\Pn-\P)\big[\varphi(O; \hat\P, \hat d)-\varphi(O; \P, d)\big].
\]
To avoid the need for Donsker conditions, we assume sample splitting: the nuisance estimates $(\hat\theta_a, \hat\pi_a)$ and selector $\hat d$ are trained on an auxiliary sample independent of the one used to evaluate $\Pn$. Conditional on this training sample, $\varphi(O;\hat P,\hat d)$ is a fixed measurable function of $O$, so that standard empirical process inequalities apply.

By standard empirical process argument, conditional on the independent training sample we have
\[
\E\big[|R_1| \mid \hat P, \hat d\big] 
\lesssim \frac{1}{\sqrt{n}} \|\varphi(O;\hat P,\hat d)-\varphi(O;P,d)\|_{L_2(P)}.
\]
Therefore, it suffices to control the $L_2(P)$–distance $\|\varphi(O;\hat P,\hat d)-\varphi(O;P,d)\|_{L_2(P)}=o_p(1)$. Similarly,
\begin{align*}
\|\varphi(O;\hat P,\hat d)-\varphi(O;P,d)\|_{L_2(P)}
&\le \|\varphi(O;\hat P,\hat d)-\varphi(O;\hat P,d)\|_{L_2(P)}
   + \|\varphi(O;\hat P,d)-\varphi(O;P,d)\|_{L_2(P)} \\
&=: T_{\mathrm{sel}} + T_{\mathrm{nuis}}.
\end{align*}

\paragraph{Bound for $T_{\mathrm{nuis}}$.}
This term corresponds to perturbations in the nuisance functions with selector fixed.
Similar to expansion as in Step 3 (see Eq.~\eqref{eq:B_nuis_bound}) but now in $L_2(P)$ norm rather than $L_1(P)$, each component is a sum of nuisance estimation errors. 
\begin{align*} 
&\varphi(O;\hat P,a)-\varphi(O;P,a) \\
=& \left[ \hat\theta_a + \frac{\1\{A=a\}}{\hat\pi_a}(\1\{Y\le y\}-\hat\theta_a) \right] - \left[ \theta_a + \frac{\1\{A=a\}}{\pi_a}(\1\{Y\le y\}-\theta_a) \right] \\
=& (\hat\theta_a - \theta_a) \left(1 - \frac{\1\{A=a\}}{\hat\pi_a} \right) + \1\{A=a\} \left( \frac{\1\{Y\le y\}-\theta_a}{\hat\pi_a} - \frac{\1\{Y\le y\}-\theta_a}{\pi_a} \right) \\
=& (\hat\theta_a - \theta_a) \left(1 - \frac{\1\{A=a\}}{\hat\pi_a} \right) + \1\{A=a\}(\1\{Y\le y\}-\theta_a) \left( \frac{\pi_a - \hat\pi_a}{\hat\pi_a \pi_a} \right) 
\end{align*}
Then
\begin{align} \label{eq:Tnuis_bound}
T_{\mathrm{nuis}} &= \left| \sum_{a=0}^1 \1\{d(X)=a\} \big(\varphi(O;\hat P,a)-\varphi(O;P,a)\big) \right|_{L_2(P)} \nonumber\\ 
&\le \sum_{a=0}^1 \left| (\hat\theta_a - \theta_a) \left(1 - \frac{\1\{A=a\}}{\hat\pi_a} \right) + \1\{A=a\}(\1\{Y\le y\}-\theta_a) \left( \frac{\pi_a - \hat\pi_a}{\hat\pi_a \pi_a} \right)  \right|_{L_2(P)} \nonumber\\
&\lesssim \sum_{a=0}^1 \left( |\hat\theta_a - \theta_a|_{L_2(P)} + |\hat\pi_a - \pi_a|_{L_2(P)} \right)=o_p(1), 
\end{align}
when both nuisance estimators are consistent $|\hat\theta_a - \theta_a|_{L_2(P)}=o_p(1)$ and $|\hat\pi_a - \pi_a|_{L_2(P)}=o_p(1)$.

\paragraph{Bound for $T_{\mathrm{sel}}$.}
Using the pointwise identity equation \ref{eq:sel_decom},
since it is supported only on the set $\{\widehat d(X)\ne d(X)\}$, taking expectation gives
\begin{align*}
T_{\mathrm{sel}}^2
&= \E\Big[\big(\varphi(O;\hat P,\hat d)-\varphi(O;\hat P,d)\big)^2\Big]
= \E\big[\Delta_d(X)^2\,\Delta_\varphi(O;\hat P)^2\big]\\
&=\E_X\left[ \Delta_d(X) \cdot \E_{A,Y|X}\left[ \Delta_\varphi(O;\hat P)^2 \mid X \right] \right], 
\end{align*}
where $\Delta_d(X):=\1\{\hat d(X)\ne d(X)\}\in\{0,1\}$, $\Delta_\varphi(O;\hat P)=\varphi(O;\hat P,1)-\varphi(O;\hat P,0)$.
Recall $\Delta_\varphi(O;\hat P)$ is bounded when $\hat\pi_a \ge \underline\pi > 0$. Consequently, 
\begin{equation}\label{eq:Tsel_bound}
T_{\mathrm{sel}} \lesssim \P(\hat d\ne d)^{1/2} \lesssim \Big(\max_a\|\hat\theta_a-\theta_a\|_\infty\Big)^{\alpha/2}=o_p(1).
\end{equation}

\paragraph{Combining.}
Substituting \eqref{eq:Tnuis_bound} and \eqref{eq:Tsel_bound} into the empirical process bound yields
\[
R_1 = o_p(n^{-1/2}).
\]
under the mild condition $\|\hat\theta_a-\theta_a\|=o_p(1)$ and $\|\hat\pi_a-\pi_a\|=o_p(1)$.

\noindent\textbf{Step 6. Asymptotic linearity and normality.}

Combining Steps 3–5, all remainders $R_1$ and $R_2$ are $o_p(n^{-1/2})$,
and thus the estimator $\hat\Psi$ in \eqref{eq:hatPsi_def} admits the asymptotic linear representation
\[
\sqrt{n}\,(\hat\Psi-\Psi)
= \frac{1}{\sqrt{n}}\sum_{i=1}^n \varphi(O_i;P,d) + o_p(1),
\]
with influence function given in \eqref{eq:IF_oracle_detailed}. Since \(\IF_{\mathrm{oracle}}(\Psi(P;d))\) has finite variance (it is a bounded combination of bounded terms under our assumptions), classical CLT gives
\[
\sqrt{n}(\hat\Psi-\Psi) \dto N\big(0,\Var(\varphi(O;P,d)\big).
\]
A consistent variance estimator is
$$\hat\sigma^2 = \Pn\left[ \left( \varphi(O; \hat P, \hat d) - \hat\Psi \right)^2 \right].$$
Under the same remainder conditions one verifies $\hat\sigma^2\convp \Var(\varphi(O;P,d))$.
\end{proof}

\subsection{Proof of asymptotic property of smooth-approximation estimator}\label{proof:smooth}
We now prove the statistical properties of smooth-approximation estimator with a fixed smooth parameter $t$ (Theorem \ref{thm:smooth}).

\begin{proof}
The proof follows standard semiparametric argument as well.
\paragraph{Step 1. Smooth approximation and differentiability.}
Define, for $t>0$,
\[
g_t(u,v) := -\frac{1}{t}\log\big(e^{-t u}+e^{-t v}\big),
\quad (u,v)\in[0,1]^2.
\]
Then $g_t(u,v)$ is smooth and satisfies
\[
\min(u,v)- \frac{\log 2}{t} \le g_t(u,v) \le \min(u,v),
\quad
\text{and }
\lim_{t\to\infty} g_t(u,v)=\min(u,v).
\]
We approximate the functional $\Psi(P)=\E[\min(\theta_0(X),\theta_1(X))]$
by
\[
\Psi_t(P)=\E\big[g_t(\theta_0(X),\theta_1(X))\big],
\]
which is continuously Gateaux-differentiable in $P$ for any finite $t$.
As $t\to\infty$, $\Psi_t(P)\uparrow\Psi(P)$.

\paragraph{Step 2. Pathwise derivative and canonical gradient.}
Our target parameter is $\Psi_t(P) = \E[g_t(\theta_0(X), \theta_1(X))]$.
Let $P_\varepsilon$ be a regular parametric submodel with score $s(O)$ at $\varepsilon=0$, the pathwise derivative of $\Psi_t(P)$ is
\[
\frac{d}{d\varepsilon}\Psi_t(P_\varepsilon)\Big|_{\varepsilon=0}
= \frac{d}{d\varepsilon} \E_\varepsilon\left[ g_t(\theta_{0,\varepsilon}(X), \theta_{1,\varepsilon}(X)) \right] \Big|_{\varepsilon=0}
\]
Using previously derived EIF for smooth functional (equation \ref{eq:IF_smooth_appendix}), we have
\begin{equation} \label{eq:IF_smooth}
\IF(\Psi_t(P)) = \sum_{a=0}^1 w_{a,t}(X) \frac{\1\{A=a\}}{\pi_a(X)}(\1\{Y\le y\}-\theta_a(X)) + g_t(\theta_0(X),\theta_1(X)) - \Psi_t(P),
\end{equation}
where $w_{a,t}(x) = \partial_a g_t(\theta_0(x), \theta_1(x))= \frac{e^{-t\theta_a(x)}}{e^{-t\theta_0(x)}+e^{-t\theta_1(x)}}$ are smooth functions bounded in $[0,1]$ with $\sum_{a=0}^1w_{a,t}(x)=1$, interpreted as a smooth weighting function between $\theta_0(x)$ and $\theta_1(x)$.

\paragraph{Step 3. Doubly-robust estimator.}
We denote $\varphi_t(O;P)$ the uncentered influence function
\[
\varphi_t(O;P) = \sum_{a=0}^1 w_{a,t}(X) \frac{\1\{A=a\}}{\pi_a(X)}(\1\{Y\le y\}-\theta_a(X)) + g_t(\theta_0(X),\theta_1(X)),
\]
and $\varphi_t(O;\hat P)$ for using estimated nuisance parameters $\hat\theta_a(X)$ and $\hat\pi_a(X)$ obtained on an independent
training sample. Define
\begin{equation}\label{eq:hatPsi_t}
\hat\Psi_t = \Pn\varphi_t(O;\hat P)
= \Pn\!\Big[
\sum_{a=0}^1 \hat w_{a,t}(X) \frac{\1\{A=a\}}{\hat\pi_a(X)}(\1\{Y\le y\}-\hat\theta_a(X)) + g_t(\hat\theta_0(X),\hat\theta_1(X))
\Big],
\end{equation}
where
\[
\hat w_{a,t}(X)
= \frac{e^{-t\hat\theta_a(X)}}{e^{-t\hat\theta_0(X)}+e^{-t\hat\theta_1(X)}}.
\]
We will then establish the rate conditions when $\hat\Psi_t$ is root-$n$ consistent and asymptotic normality.

\paragraph{Step 4. von Mises expansion and remainder decomposition.}
Let $\hat\eta$ denote estimated nuisances from the independent sample.
Then
\begin{align}
\hat\Psi_t - \Psi_t(P)
&= (\Pn-P)\varphi_t(O;P)
+ \underbrace{P[\varphi_t(O;\hat P)-\varphi_t(O;P)]}_{B_{\mathrm{nuis}}}
+ \underbrace{(\Pn-P)[\varphi_t(O;\hat P)-\varphi_t(O;P)]}_{R_n},
\label{eq:vonMises_t}
\end{align}
where $\varphi_t(O;\hat P)$ is the uncentered plug-in influence function with $\hat\eta$.

Since the training sample used for $\hat\eta$ is independent of $\Pn$,
conditional on $\hat\eta$ the term $R_n$ is a mean-zero empirical process.
We treat $B_{\mathrm{nuis}}$ and $R_n$ separately.

\paragraph{Step 5. Control of $B_{\mathrm{nuis}}$.}
The nuisance remainder is $B_{\mathrm{nuis}} = P[\varphi_t(O;\hat P)-\varphi_t(O;P)]$. We use the Law of Iterated Expectations $\E_O[\cdot] = \E_X[\E_{A,Y|X}[\cdot \mid X]]$. For the $a$-th component of the uncentered IF $\varphi_t(O,\hat P)$, the conditional expectation is
\begin{align*}
    &\E[\varphi_t(O,\hat P)|X] \\
    =& \E\left[ \sum_{a=0}^1\Big\{\hat w_{a,t}(X) \frac{\1\{A=a\}}{\hat\pi_a(X)}(\1\{Y\le y\}-\hat\theta_a(X))\Big\}+ g_t(\hat\theta_0(X),\hat\theta_1(X)) \Bigg| X \right] \\
    =& \sum_{a=0}^1\left\{\hat w_{a,t}(X) \frac{\E[\1\{A=a\} \mid X]}{\hat\pi_a(X)} \left( \E[\1\{Y\le y\} \mid A=a, X] - \hat\theta_a(X) \right)\right\}+ g_t(\hat\theta_0(X),\hat\theta_1(X)) \\
    =& \sum_{a=0}^1\left\{\hat w_{a,t}(X) \frac{\pi_a(X)}{\hat\pi_a(X)} \left( \theta_a(X) - \hat\theta_a(X) \right)\right\}+ g_t(\hat\theta_0(X),\hat\theta_1(X)).
\end{align*}
Similarly, the conditional expectation of the $\varphi_t$ evaluated at the true parameters $P$ is
\begin{align*}
    \E[\varphi_t(O;P) \mid X]
    &= \sum_{a=0}^1 w_{a,t}(X) \frac{\pi_a(X)}{\pi_a(X)} \left( \theta_a(X) - \theta_a(X) \right) + g_t(\theta_0(X), \theta_1(X)) \\
    &= g_t(\theta_0(X), \theta_1(X)).
\end{align*}
Therefore, $B_{\mathrm{nuis}}$ is given by the expectation of the difference in conditional means
$$
B_{\mathrm{nuis}} = \E_X\left[ \E[\varphi_t(O;\hat P) \mid X] - \E[\varphi_t(O;P) \mid X] \right],
$$
where
\begin{align*}
\E[\varphi_t(O;\hat P) \mid X] - \E[\varphi_t(O;P) \mid X] &= \left[ \sum_{a=0}^1 \hat w_{a,t} \frac{\pi_a}{\hat\pi_a} (\theta_a - \hat\theta_a) + g_t(\hat\theta_0, \hat\theta_1) \right] - g_t(\theta_0, \theta_1) \\
    &= \sum_{a=0}^1 \hat w_{a,t} \frac{\pi_a}{\hat\pi_a} (\theta_a - \hat\theta_a) - \left[ g_t(\theta_0, \theta_1) - g_t(\hat\theta_0, \hat\theta_1) \right].
\end{align*}
Now, we use the first-order Taylor expansion for $g_t(\theta_0, \theta_1)$ around $\hat\theta$,
$$
g_t(\theta_0, \theta_1) - g_t(\hat\theta_0, \hat\theta_1) = \sum_{a=0}^1 \hat w_{a,t} (\theta_a - \hat\theta_a) + R_{\theta},
$$
where $R_{\theta}$ is a quadratic remainder term $|R_\theta(X)| \lesssim t \cdot \|\hat\theta(X) - \theta(X)\|_2^2 = t \cdot \sum_{a=0}^1 (\hat\theta_a(X) - \theta_a(X))^2$, since $ \frac{\partial^2 g_t}{\partial \theta_0^2} = -t \cdot \frac{e^{-t\theta_0} e^{-t\theta_1}}{(e^{-t\theta_0}+e^{-t\theta_1})^2} = -t \cdot w_{0,t} \cdot w_{1,t}$.
Substituting this expansion back
\begin{align*}
    \E[\varphi_t(O;\hat P) \mid X] - \E[\varphi_t(O;P) \mid X] &= \sum_{a=0}^1 \hat w_{a,t} \frac{\pi_a}{\hat\pi_a} (\theta_a - \hat\theta_a) - \left[ \sum_{a=0}^1 \hat w_{a,t} (\theta_a - \hat\theta_a) + R_{\theta} \right] \\
    &= \sum_{a=0}^1 \hat w_{a,t} (\theta_a - \hat\theta_a) \left( \frac{\pi_a}{\hat\pi_a} - 1 \right) - R_{\theta} \\
    &= \sum_{a=0}^1 \hat w_{a,t} (\hat\theta_a - \theta_a) \left( \frac{\hat\pi_a - \pi_a}{\hat\pi_a} \right) - R_{\theta}.
\end{align*}
The first term is the key second-order interaction term. The remainder $R_{\theta}$ is $O_p(t\|\hat\theta-\theta\|_{L_2(P)}^2)$.
The term $\hat w_{a,t}/\hat\pi_a$ is bounded almost surely by $\sup |\hat w_{a,t}| / \inf \hat\pi_a \le 1/\underline\pi$.
Thus, $B_{\mathrm{nuis}}$
$$
|B_{\mathrm{nuis}}| \lesssim \sum_{a=0}^1 \E_X\left[ |\hat\theta_a(X)-\theta_a(X)| \cdot |\hat\pi_a(X)-\pi_a(X)| \right] + O_p\left(t \cdot \|\hat\theta-\theta\|_{L_2(P)}^2 \right).
$$
Applying the Cauchy-Schwarz inequality, the first term is bounded by $\sum_{a=0}^1 \|\hat\theta_a-\theta_a\|_{L_2(P)} \cdot \|\hat\pi_a-\pi_a\|_{L_2(P)}$. We get
\begin{equation}\label{eq:smooth_nuis}
    |B_{\mathrm{nuis}}| \lesssim O_p(\|\hat\theta-\theta\|_{L_2(P)} \cdot \|\hat\pi-\pi\|_{L_2(P)} + t\|\hat\theta-\theta\|_{L_2(P)}^2).
\end{equation}
For a fixed $t$, under condition
$$
\|\hat\theta_a-\theta_a\|_{L_2(P)}\|\hat\pi_a-\pi_a\|_{L_2(P)}=o_p(n^{-1/2}),\quad\|\hat\theta_a-\theta_a\|^2_{L_2(P)} = o_p(n^{-1/2}),
$$
for example $\|\hat\theta_a-\theta_a\|_{L_2(P)}=o_p(n^{-1/4})$, $\|\hat\pi_a-\pi_a\|_{L_2(P)} = o_p(n^{-1/4})$, we conclude that the nuisance remainder term vanishes faster than the $\sqrt{n}$ rate $B_{\mathrm{nuis}} = o_p(n^{-1/2})$.

\paragraph{Step 6. Bound for Empirical Process Remainder $R_n$}
The empirical process remainder is $R_n = (\Pn-P)[\varphi_t(O;\hat P)-\varphi_t(O;P)]$.
Let the difference function be $\Delta\varphi_t(O) = \varphi_t(O;\hat P)-\varphi_t(O;P)$. We aim to show $\sqrt{n} R_n = o_p(1)$.

Since the nuisance parameters $\hat\eta$ are obtained on a sample independent of the evaluation sample used for $\Pn$ (sample splitting), we use the conditional $\sqrt{n}$ concentration bound
$$
\E[R_n^2 \mid \hat\eta] \le \frac{1}{n} \E\left[ (\Delta\varphi_t(O))^2 \mid \hat\eta \right] = \frac{1}{n} \|\Delta\varphi_t\|_{L_2(P)}^2,
$$
where $\|\cdot\|_{L_2(P)}$ denotes the $L_2(P)$ norm of the function $\Delta\varphi_t(O)$ given $\hat\eta$.

The function $\varphi_t(O;\eta)$ is a smooth function of $\eta=(\theta_0, \theta_1, \pi_0, \pi_1)$. Given the assumptions that $\theta_a$ and $\hat\theta_a$ are bounded in $[0,1]$ and $\pi_a$ and $\hat\pi_a$ are bounded away from zero (i.e., $\inf \hat\pi_a > \underline\pi > 0$ a.s.), $\varphi_t$ is locally Lipschitz continuous in $\eta$
\begin{align*}
    |\Delta\varphi_t(O)| &\lesssim \left( \sum_{a=0}^1 \sup_{\eta} \left|\frac{\partial \varphi_t}{\partial \theta_a}\right| \cdot |\hat\theta_a - \theta_a| + \sum_{a=0}^1 \sup_{\eta} \left|\frac{\partial \varphi_t}{\partial \pi_a}\right| \cdot |\hat\pi_a - \pi_a| \right)\\
    & \lesssim \left( O(t) \cdot \sum_{a=0}^1 |\hat\theta_a(X) - \theta_a(X)| + O(1) \cdot \sum_{a=0}^1 |\hat\pi_a(X) - \pi_a(X)| \right) \cdot K(O)
\end{align*}
where $K(O)$ is a bounded function depending on $A, Y$ and the bounding constants for $\hat\pi_a$. Squaring the difference and taking the expectation $P$
\begin{align*}
    \|\Delta\varphi_t\|_{L_2(P)}^2 &= \E_O\left[ (\Delta\varphi_t(O))^2 \right] \\
    &\lesssim t^2 \|\hat\theta - \theta\|_{L_2(P)}^2 + \|\hat\pi - \pi\|_{L_2(P)}^2.
\end{align*}
Substituting this bound back into the variance of $R_n$
$$
\E[R_n^2 \mid \hat\eta] \lesssim \frac{1}{n} \left[t^2 \|\hat\theta - \theta\|_{L_2(P)}^2 + \|\hat\pi - \pi\|_{L_2(P)}^2\right].
$$
Taking the square root and applying the Markov inequality yields
$$
R_n = O_p\left(n^{-1/2} \left[t \|\hat\theta - \theta\|_{L_2(P)} + \|\hat\pi - \pi\|_{L_2(P)}\right]\right).
$$
Then for a fixed $t$, we have $R_n=o_p(n^{-1/2})$ under the condition $\|\hat\theta-\theta\|_{L_2(P)} = o_p(1)$, $\|\hat\pi-\pi\|_{L_2(P)} = o_p(1)$.

\paragraph{Step 7. Asymptotic linearity and normality.}
Combining the bounds above, under $
\|\hat\theta_a-\theta_a\|_{L_2(P)} \cdot \|\hat\pi_a-\pi_a\|_{L_2(P)} = o_p(n^{-1/2}),
$, we obtain the asymptotic expansion
\[
\hat\Psi_t - \Psi_t(P)
= (\Pn-P)\varphi_t(O;P) + o_p(n^{-1/2}),
\]
and hence the central limit theorem
\[
\sqrt{n}(\hat\Psi_t-\Psi_t(P)) \dto N(0,\Var(\varphi_t(O;P))).
\]
A consistent variance estimator is
\[
\hat\sigma_t^2 = \Pn\big[\big(\varphi_t(O;\hat P) - \hat\Psi_t\big)^2\big].
\]  
\end{proof}

\section{Proof of identifability of triple matching learning estimator (Theorem \ref{thm:rep-iv-ate}).}\label{proof:rep-iv-ate}
We prove the identifiability results in Theorem \ref{thm:rep-iv-ate}.
\begin{proof}
We first establish the identifability of confounding representation $Z_C$, and then use learned $Z_C$ to identify potential outcomes $Y(a)$.

\paragraph{Step 1: Identifiability of the Confounding Subspace.}
Let the true data generating process be defined by the structural equations $A = g_A(Z_S)$ and $Y = g_Y(A, Z_C)$. We do not assume $g_A$ is injective, but we assume that for each fixed $a$, the conditional distribution $p(Y \mid A=a, Z_C=\cdot)$ is injective with respect to $Z_C$ (Assumption~\ref{assump:A1}). (Note that if $Y$ is deterministic, this reduces to $g_Y(a, \cdot)$ being an injective function).

Let $\hat{e} : \mathcal{A}\times\mathcal{Y} \to \mathcal{Z}_{\widehat{Z}_C}$ be the learned encoder, defining $\widehat{Z}_C = \hat{e}(A, Y)$. Define the composite map
$$\psi(z_c, z_s) := \hat{e}\!\left(g_A(z_s),\; g_Y(g_A(z_s), z_c)\right).$$
We employ the following regularity conditions
\begin{enumerate}
\item[(i)] \textbf{Independence Constraint:} The learned representation satisfies $\widehat{Z}_C \perp S$.
\item[(ii)] \textbf{Predictive Sufficiency:} The learned representation is sufficient for predicting $Y$ from $A$. That is, $Z_C$ provides no additional information about $Y$ once $\widehat{Z}_C$ is known
$$Y \perp\!\!\!\perp Z_C \mid (A, \widehat{Z}_C) \quad \text{almost surely}.$$
In terms of densities, this implies $p(Y \mid A, \widehat{Z}_C, Z_C) = p(Y \mid A, \widehat{Z}_C)$ almost surely.
\item[(iii)] \textbf{Sufficient Variability:} As in Assumption~\ref{assump:variability}, the family of conditional distributions $\{p(Z_S \mid Z_C=z_c, S=s)\}_{s}$ is boundedly complete.
\end{enumerate}

\paragraph{1(a) Functional independence ($\widehat{Z}_C$ depends only on $Z_C$).}
Fix any measurable set $U \subseteq \mathcal{Z}_{\widehat{Z}_C}$ and let $D = \psi^{-1}(U)$ be its preimage. By condition (i), $P(\psi(Z) \in U \mid S=s)$ is invariant to $s$. Using the factorization $p(z \mid s)=p(z_c)\,p(z_s\mid z_c,s)$, we obtain the identity
\begin{equation}\label{eq:proof_integral_final}\int p(z_c)\bigl[p(z_s \mid z_c, s_1) - p(z_s \mid z_c, s_2)\bigr] \mathbf{1}_D(z_c,z_s), dz_s dz_c = 0.\end{equation}
Suppose for contradiction that $\psi$ depends on $z_s$. Then $D$ is not a Cartesian product almost surely. Let $B^* = \{(z_c,z_s)\in D : \{z_c\}\times\mathcal{Z}_S \not\subseteq D\}$ be the entangled region, which has positive measure. By the completeness condition (iii), the integral over this non-product region cannot vanish for all $s_1, s_2$, contradicting \eqref{eq:proof_integral_final}.

Thus, $D = B \times \mathcal{Z}_S$ almost surely, implying $\psi(z_c, z_s)$ is constant in $z_s$. Hence, there exists a measurable map $\phi : \mathcal{Z}_C \to \mathcal{Z}_{\widehat{Z}_C}$ such that $\widehat{Z}_C = \phi(Z_C)$ almost surely.

\paragraph{1(b) Injectivity via Predictive Sufficiency.}
We show $\phi$ is injective by contradiction. Suppose $\phi$ is not injective. Then there exist disjoint sets $\mathcal{Z}_1, \mathcal{Z}_2 \subset \mathcal{Z}_C$ with positive measure such that $\phi(\mathcal{Z}_1) = \phi(\mathcal{Z}_2) = \hat{z}$. Fix a treatment $a \in \mathcal{A}$. By the injectivity of the generative mechanism (Assumption~\ref{assump:A1}), the outcome distributions conditioned on distinct latent values must differ. Thus, for any $z_1 \in \mathcal{Z}_1$ and $z_2 \in \mathcal{Z}_2$,
$$p(Y \mid A=a, Z_C=z_1) \neq p(Y \mid A=a, Z_C=z_2).$$
(Note: If $Y$ is deterministic, these are distinct Dirac measures $\delta_{y_1} \neq \delta_{y_2}$).

Now consider the distribution conditioned on $\widehat{Z}_C = \hat{z}$. By the \textbf{Predictive Sufficiency} condition (ii), we have conditional independence $Y \perp Z_C \mid (A, \widehat{Z}_C)$. This implies that
$$p(Y \mid A=a, \widehat{Z}_C=\hat{z}, Z_C=z) = p(Y \mid A=a, \widehat{Z}_C=\hat{z})$$
for almost all $z$ in the fiber $\phi^{-1}(\hat{z})$.

However, the left-hand side $p(Y \mid A=a, \widehat{Z}_C=\hat{z}, Z_C=z)$ simplifies because $A$ and $Z_C$ fully determines the true conditional distribution of $Y$. Thus, for $z_1 \in \mathcal{Z}_1$ and $z_2 \in \mathcal{Z}_2$, we must have
$$p(Y \mid A=a, Z_C=z_1) = p(Y \mid A=a, \widehat{Z}_C=\hat{z}) = p(Y \mid A=a, Z_C=z_2).$$
This equates two distributions that are known to be distinct (due to injectivity), yielding a contradiction. Therefore, $\phi$ must be injective almost surely. Under standard regularity conditions (continuity and matching dimensions), $\phi$ is an invertible transformation.
\medskip
We have established that $\widehat{Z}_C = \phi(Z_C)$ where $\phi$ is invertible, which implies $\sigma(\widehat{Z}_C) = \sigma(Z_C)$. This completes the proof of subspace identifiability of $Z_C$.

\paragraph{Step 2: Back-door Criterion.}
Given the structural equations in Assumption \ref{assump:A1}:
\[ A = g_A(Z_S), \quad Y = g_Y(A, Z_C), \quad Z_S = h(Z_C, S, \varepsilon_S). \]
The only common cause of $A$ and $Y$ is $Z_C$ (mediated through $Z_S$ to $A$). $Z_C$ blocks all back-door paths from $A$ to $Y$. Specifically, the potential outcome $Y(a)$ is determined by $g_Y(a, Z_C)$. Since $Z_C$ encapsulates all confounding information, we have conditional exchangeability
\[
Y(a) \indep A \mid Z_C.
\]

\paragraph{Step 3: Replacement with Identified Representation.}
From Step 1, we established that $\widehat{Z}_C = \psi(Z_C)$ where $\psi$ is invertible. Thus, $\sigma(\widehat{Z}_C) = \sigma(Z_C)$, and conditioning on $\widehat{Z}_C$ is statistically equivalent to conditioning on $Z_C$. Therefore, 
\[
Y(a) \indep A \mid \widehat{Z}_C.
\]

\paragraph{Step 4: Identification Formula.}
We explicitly derive the identification of $\mathbb{E}[Y(a)]$. By the Law of Iterated Expectations and the independence shown in Step 3
\begin{align*}
    \mathbb{E}[Y(a)] &= \mathbb{E}_{\widehat{Z}_C} \big[ \mathbb{E}[Y(a) \mid \widehat{Z}_C] \big] \\
    &= \mathbb{E}_{\widehat{Z}_C} \big[ \mathbb{E}[Y(a) \mid A=a, \widehat{Z}_C] \big] \quad (\text{by ignorability } Y(a) \indep A \mid \widehat{Z}_C) \\
    &= \mathbb{E}_{\widehat{Z}_C} \big[ \mathbb{E}[Y \mid A=a, \widehat{Z}_C] \big] \quad (\text{by consistency } Y(a) = Y \text{ when } A=a).
\end{align*}
Assumption \ref{assump:A3} (Positivity) guarantees that $P(A=a \mid \widehat{Z}_C) > 0$ (since $\widehat{Z}_C$ is isomorphic to $Z_C$), ensuring the conditional expectation $\mathbb{E}[Y \mid A=a, \widehat{Z}_C]$ is well-defined.
Similarly, the ATE is identified as
\[
\mathrm{ATE}(a,a') = \mathbb{E}_{\widehat{Z}_C} \left[ \mathbb{E}[Y \mid A=a, \widehat{Z}_C] - \mathbb{E}[Y \mid A=a', \widehat{Z}_C] \right].
\]
This completes the proof.
\end{proof}

\section{Proof of statistical properties of triple machine learning estimator (Theorem \ref{thm:dr-rliv-summary})}\label{proof:dr-rliv-summary}
\begin{proof}
Let the total sample size be $N$. We randomly partition the data $\mathcal{D}$ into three disjoint folds $\mathcal{I}_1, \mathcal{I}_2, \mathcal{I}_3$, each of size $n = N/3$.
The estimator is constructed sequentially
\begin{enumerate}
    \item \textbf{Stage 1 (Representation Learning on $\mathcal{I}_1$):} 
    Construct $\widehat{\phi}$ using only data in $\mathcal{I}_1$. Thus $\widehat{\phi} \indep (\mathcal{I}_2 \cup \mathcal{I}_3)$.
    
    \item \textbf{Stage 2 (Nuisance Estimation on $\mathcal{I}_2$):} 
    Using $\widehat{\phi}$ and data $\mathcal{I}_2$, estimate $\widehat{m}$ and $\widehat{\pi}$. Let $\widehat{\eta} = (\widehat{m}, \widehat{\pi}, \widehat{\phi})$. Crucially, $\widehat{\eta} \indep \mathcal{I}_3$.
    
    \item \textbf{Stage 3 (Evaluation on $\mathcal{I}_3$):} 
    Compute the estimator on $\mathcal{I}_3$
    \[
    \widehat{\psi} = \mathbb{P}_{n,3} [\varphi(O; \widehat{\eta})] = \frac{1}{n} \sum_{i \in \mathcal{I}_3} \varphi(O_i; \widehat{\eta}).
    \]
\end{enumerate}

We decompose the estimation error $\sqrt{n}(\widehat{\psi} - \psi_0)$
\begin{align}
\sqrt{n}(\widehat{\psi} - \psi_0) &= \sqrt{n}(\mathbb{P}_{n,3} \varphi(\widehat{\eta}) - P \varphi(\eta_0)) \nonumber \\
&= \underbrace{\sqrt{n}(\mathbb{P}_{n,3} - P) \varphi(\eta_0)}_{T_1: \text{Oracle CLT}} 
+ \underbrace{\sqrt{n}(\mathbb{P}_{n,3} - P) (\varphi(\widehat{\eta}) - \varphi(\eta_0))}_{T_2: \text{Empirical Process}} 
+ \underbrace{\sqrt{n} P (\varphi(\widehat{\eta}) - \varphi(\eta_0))}_{T_3: \text{Bias Term}}.
\label{eq:decomp}
\end{align}

\paragraph{Step 1: Oracle CLT ($T_1$).}
The term $\varphi(O; \eta_0)$ is a fixed function. Since observations in $\mathcal{I}_3$ are i.i.d., by the standard CLT
\[
T_1 \xrightarrow{d} \mathcal{N}(0, \sigma^2_\text{eff}), \quad \text{where } \sigma^2_\text{eff} = \mathrm{Var}(\varphi(O; \eta_0)).
\]

\paragraph{Step 2: Empirical Process ($T_2$).}
We must show that $T_2 = o_p(1)$. This is equivalent to showing that the unscaled empirical process term $(\mathbb{P}_{n,3} - P)(\varphi(\widehat{\eta}) - \varphi(\eta_0))$ is $o_p(n^{-1/2})$.

Let $\Delta(O; \widehat{\eta}) = \varphi(O; \widehat{\eta}) - \varphi(O; \eta_0)$. 
Conditioning on the training data $\mathcal{D}_{train} = \mathcal{I}_1 \cup \mathcal{I}_2$, the function $\Delta(\cdot; \widehat{\eta})$ is deterministic. The term $T_2$ can be viewed as a sum of i.i.d. random variables with mean zero (conditional on $\mathcal{D}_{train}$)
\[
\E[T_2 \mid \mathcal{D}_{train}] = \sqrt{n} \E_{O \sim P} [ (\Pn - P)\Delta(O; \widehat{\eta}) \mid \mathcal{D}_{train} ] = 0.
\]
We analyze the conditional variance
\begin{align*}
\mathrm{Var}(T_2 \mid \mathcal{D}_{train}) &= n \cdot \mathrm{Var}(\Pn \Delta(O; \widehat{\eta}) \mid \mathcal{D}_{train}) \\
&= n \cdot \frac{1}{n} \mathrm{Var}(\Delta(O; \widehat{\eta}) \mid \mathcal{D}_{train}) \\
&\leq \E [ (\varphi(O; \widehat{\eta}) - \varphi(O; \eta_0))^2 \mid \mathcal{D}_{train} ] \\
&= \| \varphi(\widehat{\eta}) - \varphi(\eta_0) \|_{L_2(P)}^2.
\end{align*}
Under the consistency assumption (C1), we have $\|\widehat{\eta} - \eta_0\| \xrightarrow{p} 0$. Assuming $\varphi$ satisfies mild Lipschitz continuity or the nuisances are bounded, consistency implies convergence in the $L_2$ norm of the score
\[
\| \varphi(\widehat{\eta}) - \varphi(\eta_0) \|_{L_2(P)}^2 = o_p(1).
\]
By Chebyshev's inequality, for any $\epsilon > 0$
\[
P(|T_2| > \epsilon \mid \mathcal{D}_{train}) \leq \frac{\mathrm{Var}(T_2 \mid \mathcal{D}_{train})}{\epsilon^2} = \frac{o_p(1)}{\epsilon^2} \xrightarrow{p} 0.
\]
Thus, $T_2 = o_p(1)$. This confirms that the estimation noise of $\widehat{\eta}$ does not affect the asymptotic distribution via the empirical process term.

\paragraph{Step 3: Bias Term ($T_3$).}
We analyze the drift term $T_3 = \sqrt{n} \E[\varphi(O; \widehat{\eta}) - \varphi(O; \eta_0) \mid \widehat{\eta}]$. 
Define the intermediate parameter $\tilde{\eta} = (m_0, \pi_0, \widehat{\phi})$, which represents the ideal nuisance parameters given the learned representation $\widehat{\phi}$. 
We decompose the bias into a nuisance estimation error and a representation learning error
\[
T_3 = \underbrace{\sqrt{n} P (\varphi(\widehat{\eta}) - \varphi(\tilde{\eta}))}_{T_{3a} \text{ (Nuisance Bias)}} + \underbrace{\sqrt{n} P (\varphi(\tilde{\eta}) - \varphi(\eta_0))}_{T_{3b} \text{ (Representation Bias)}}.
\]

\textbf{(a) Nuisance Parameter Bias ($T_{3a}$).}
This term captures the error from estimating $m$ and $\pi$ on $\mathcal{I}_2$, conditional on the fixed representation $\widehat{\phi}$ from $\mathcal{I}_1$. Utilizing the algebraic property of the doubly robust score function, for any $m, \pi$ and fixed representation $z$, the difference satisfies the exact identity
\begin{align*}
\E[\varphi(m, \pi, z) - \varphi(m_0, \pi_0, z)] &= \E \left[ \frac{\1\{A=a\}}{\pi(z)}(Y - m(z)) - \frac{\1\{A=a\}}{\pi_0(z)}(Y - m_0(z)) + (m(z) - m_0(z)) \right] \\
&= - \E \left[ \frac{\pi(z) - \pi_0(z)}{\pi(z)} \big( m(z) - m_0(z) \big) \right].
\end{align*}
Applying this to our estimator $\widehat{\eta}$ given $\widehat{\phi}$
\[
T_{3a} = - \sqrt{n} \E \left[ \frac{(\widehat{\pi}(\widehat{Z}_C) - \pi_0(\widehat{Z}_C))(\widehat{m}(\widehat{Z}_C) - m_0(\widehat{Z}_C))}{\widehat{\pi}(\widehat{Z}_C)} \Bigg| \mathcal{I}_1 \right].
\]
Note that the first-order terms vanish identically due to Neyman orthogonality. The remaining term is strictly second-order. By the Cauchy-Schwarz inequality and the boundedness of $1/\widehat{\pi}$
\[
|T_{3a}| \lesssim \sqrt{n} \| \widehat{\pi} - \pi_0 \|_{L_2(\widehat{\phi})} \| \widehat{m} - m_0 \|_{L_2(\widehat{\phi})}.
\]
Under the robustness assumption (product of rates is $o_p(n^{-1/2})$), we have $T_{3a} = o_p(1)$.

\textbf{(b) Representation Bias ($T_{3b}$).}
This term captures the stochastic error propagated from the representation learning step (Stage 1) to the final estimation (Stage 3). Let $M(\phi) = \E[\varphi(O; m_0, \pi_0, \phi)]$ be the expected score functional. Since the score $\varphi$ is generally \emph{not} orthogonal with respect to $\phi$, we perform a functional Taylor expansion around the true representation $\phi_0$, 
\[
T_{3b} = \sqrt{n} (M(\widehat{\phi}) - M(\phi_0)) = \underbrace{\sqrt{n} \nabla_{\phi} M(\phi_0)[\widehat{\phi} - \phi_0]}_{\text{Linear Term (I)}} + \underbrace{\sqrt{n} \mathcal{R}(\widehat{\phi}, \phi_0)}_{\text{Remainder Term (II)}},
\]
where $\nabla_{\phi} M(\phi_0)[h]$ is the Gâteaux derivative of $M$ in direction $h$.

The remainder Term (II) is bounded by the square of the estimation error $|\mathcal{R}| \le C \|\widehat{\phi} - \phi_0\|_{L_2}^2$.
For the remainder to be asymptotically negligible (i.e., $o_p(1)$), we only require the \textbf{quarter-rate condition}
\[
\|\widehat{\phi} - \phi_0\|_{L_2} = o_p(n^{-1/4}).
\]
Assuming this holds, the asymptotic behavior of $T_{3b}$ is entirely determined by the Linear Term (I). We consider two regimes

\begin{itemize}
    \item \textbf{Case 1: Super-Efficiency (Oracle Variance).}
    Suppose the representation is learned on a massive auxiliary dataset or converges strictly faster than the parametric rate
    \[
    \|\widehat{\phi} - \phi_0\|_{L_2} = o_p(n^{-1/2}).
    \]
    Then, the Linear Term (I) satisfies
    \[
    |\sqrt{n} \nabla_{\phi} M(\phi_0)[\widehat{\phi} - \phi_0]| \le C \sqrt{n} \|\widehat{\phi} - \phi_0\|_{L_2} = \sqrt{n} \cdot o_p(n^{-1/2}) = o_p(1).
    \]
    The representation bias vanishes. The estimator achieves the \textbf{Oracle Efficiency Bound}, with asymptotic variance $V_\text{eff} = \text{Var}(\varphi(O; \eta_0))$.

    \item \textbf{Case 2: Standard Rate (Variance Inflation).}
    Suppose the representation is learned at the standard parametric rate (e.g., via regression or standard ML on Fold 1)
    \[
    \|\widehat{\phi} - \phi_0\|_{L_2} = O_p(n^{-1/2}).
    \]
    In this case, assume that $\widehat{\phi}$ admits an asymptotic linear expansion characterized by its own influence function $\xi_{\phi}$
    \[
    \widehat{\phi}(z) - \phi_0(z) = \frac{1}{n_1} \sum_{j \in \mathcal{I}_1} \xi_{\phi}(O_j, z) + o_p(n^{-1/2}).
    \]
    and then define
    \[
    \mathrm{IF}_{\phi, \text{rep}}(O) := \langle \nabla_\phi M(\phi_0), \xi_\phi(O) \rangle = \mathbb{E}_{O'}[\nabla_\phi \varphi(O'; \eta_0)] \cdot \xi_\phi(O).
    \]
    By the linearity of the derivative, 
    \[
    \sqrt{n} P (\varphi(\tilde{\eta}) - \varphi(\eta_0)) \simeq \frac{1}{\sqrt{n}} \sum_{i=1}^n \mathrm{IF}_{\phi, \text{rep}}(O_i),
    \]
    and the Linear Term (I) converges in distribution
    \[
    \sqrt{n} \nabla_{\phi} M(\phi_0)[\widehat{\phi} - \phi_0] \xrightarrow{d} \mathcal{N}(0, V_{\text{rep}}),
    \]
    where $V_{\text{rep}}=\Var(\mathrm{IF}_{\phi, \text{rep}}(O))$ is the variance contribution from the representation learning step.
    
    Because $\widehat{\phi}$ is estimated on $\mathcal{I}_1$ and the evaluation score $\varphi$ is computed on the independent fold $\mathcal{I}_3$, the error term $\nabla_{\phi} M(\phi_0)[\widehat{\phi} - \phi_0]$ and the oracle influence function $\mathrm{IF}_\text{oracle}(O_i)$ are uncorrelated. This justifies the decoupled summation of variances in $V_\text{total}$. The total asymptotic variance hence inflates to
    \[
    V_\text{total} = V_\text{eff} + \rho \cdot V_{\text{rep}},
    \]
    where $\rho$ accounts for the ratio of sample sizes between folds. Standard errors must be corrected to account for $V_{\text{rep}}$.
\end{itemize}

On thing we need to emphasize is that we acknowledge that establishing the exact asymptotic linearity for highly non-convex deep learning models like VAEs remains an open theoretical challenge. Our derivation of $\mathrm{IF}_{\phi, \text{rep}}$ operates under the premise that the representation learner converges to an isolated local optimum, behaving asymptotically as a regularized M-estimator, or alternatively, operates in a regime where the neural tangent kernel (NTK) affords linear responsiveness \cite{jacot2018neural}.

Combining the steps, if the bias terms vanish ($o_p(1)$), we have $\sqrt{n}(\widehat{\psi} - \psi_0) = T_1 + o_p(1) \xrightarrow{d} \mathcal{N}(0, \sigma^2_\text{eff})$.
\end{proof}